\documentclass[english, 12pt]{article}
\usepackage[bookmarks=true, colorlinks=true, linkcolor=blue, citecolor=blue, breaklinks=true]{hyperref}
\usepackage{natbib}
\usepackage[T1]{fontenc}
\usepackage[latin9]{inputenc}
\usepackage[active]{srcltx}
\usepackage{booktabs}
\usepackage{babel}
\usepackage{float}
\usepackage{amsthm}
\usepackage{amsmath}
\usepackage{amssymb}
\usepackage{graphicx, graphics}
\usepackage{setspace}
\usepackage{esint}
\usepackage{algorithm}
\usepackage{subfig}
\usepackage[usenames]{color}
\usepackage[normalem]{ulem}
\definecolor{mygreen}{rgb}{0,.5,0}

\newcommand{\M}{{\mathcal{M}}}
\newcommand{\cS}{{\mathcal{S}}}

\newcommand{\R}{{\mathbb{R}}}

\newcommand{\be}{\begin{equation}}
\newcommand{\ee}{\end{equation}}
\newcommand{\bee}{\begin{equation*}}
\newcommand{\eee}{\end{equation*}}

\newtheorem{thm}{\protect\theoremname}

\theoremstyle{plain}
  
  \theoremstyle{plain}

  \theoremstyle{plain}
  \newtheorem*{assumption*}{\protect\assumptionname}
  \theoremstyle{plain}

  \theoremstyle{remark}
  \newtheorem*{rem*}{\protect\remarkname}
  \theoremstyle{plain}
\usepackage{fullpage}
 \usepackage{pdfsync}
\usepackage[title,titletoc]{appendix}

  \providecommand{\assumptionname}{Assumption}
  \providecommand{\lemmaname}{Lemma}
  \providecommand{\propositionname}{Proposition}
  \providecommand{\remarkname}{Remark}
\providecommand{\theoremname}{Theorem}
\providecommand{\corollaryname}{Corollary}

\begin{document}

\title{EM Algorithm and Stochastic Control in Economics}

\date{November 6, 2016}

\author{Steven Kou\thanks{Risk Management Institute and Department of Mathematics, National University of Singapore, 21 Heng Mui Keng Terrace, Singapore. Email: matsteve@nus.edu.sg.}  \and Xianhua Peng\thanks{Department of Mathematics, Hong Kong University of Science and Technology, Clear Water Bay, Kowloon, Hong Kong. Email: maxhpeng@ust.hk.} \and Xingbo Xu\thanks{Department of Industrial Engineering and Operations Research, Columbia University, New York, New York 10027, USA. Email: xx2126@columbia.edu.}}

\maketitle

\begin{abstract}
Generalising the idea of the classical EM algorithm that is widely used for computing maximum likelihood estimates, we propose an EM-Control (EM-C) algorithm for solving multi-period finite time horizon stochastic control problems. The new algorithm sequentially updates the control policies in each time period using Monte Carlo simulation in a forward-backward manner; in other words, the algorithm goes forward in simulation and backward in optimization in each iteration.
Similar to the EM algorithm, the EM-C algorithm has the monotonicity of performance improvement in each iteration, leading to good convergence properties.
We demonstrate the effectiveness of the algorithm by solving stochastic control problems in the monopoly pricing of perishable assets and in the study of real business cycle.
%The new algorithm extends the existing literature as follows: %(2) If the optimal policy can be  spanned by the basis functions, then the EM-C algorithm will converge in one iteration;
%(2) The EM-C algorithm applies to general stochastic control problems that may not satisfy the dynamic programming principle.
%(3) We do not assume any particular dynamics of the stochastic processes such as diffusion or jump diffusions.



% C44	Operations Research, Statistical Decision Theory
% C61	Optimization Techniques,  Programming Models, Dynamic Analysis
% C63	Computational Techniques, Simulation Modeling
% D4	Market Structure, Pricing, and Design, D42	Monopoly
% E3	Prices, Business Fluctuations, and Cycles, E32	Business Fluctuations  Cycles

% MSC Class:　93E20, 93E35, 91G60, 91G80, 90B05, 90C15, 90C35, 90C39, 90C40

\end{abstract}

\baselineskip 18pt

\section{Introduction}%\label{sec:introduction}

\subsection{Motivation and Main Results}

Stochastic control problems are widely used in macroeconomics (e.g., the study of real business cycle), microeconomics (e.g., utility maximization problem), and marketing (e.g., monopoly pricing of perishable assets). These control problems are likely to be of finite time horizon. However, a finite time horizon stochastic control problem is more difficult than the related infinite horizon problem, because the optimal control policy is not necessarily stationary. Usually one has to resort to numerical methods to find solutions for such finite time horizon stochastic control problems. Due to the curse of dimensionality, it is generally difficult to numerically solve such problems, especially in high dimension and for complicated stochastic dynamics.

To overcome these difficulties, in this paper we attempt to solve finite time horizon stochastic control problems by using Monte Carlo simulation.  More precisely, we propose a new algorithm, EM-Control (EM-C) algorithm, that sequentially updates the control policies in each time period using Monte Carlo simulation in a forward-backward manner; in other words, the algorithm goes forward in simulation and backward in optimization in each iteration.
We demonstrate the effectiveness of the algorithm by solving stochastic control problems in the monopoly pricing of perishable assets and in the study of real business cycle.

Our algorithm is motivated from an algorithm in a different field, the classical Expectation-Maximization (EM) algorithm (\cite*{Dempster1977}), which is widely used for computing maximum likelihood estimates (MLEs) for missing data or latent variables.
%The EM algorithm is an iterative method widely used for parameter
%estimation in maximum likelihood estimation (MLE) problems.
In each iteration, the EM algorithm first calculates the conditional distribution of the missing data based on parameters from the
previous iteration, and then maximizes the expectation of the full likelihood function based on the just updated conditional distribution to get
updated parameters. Interestingly, the EM algorithm can be viewed as an algorithm that in each iteration alternatively maximizes an objective functional with one distribution parameter and one ordinary parameter: the distribution parameter is the conditional distribution of the missing data, and the ordinary parameter is the parameter of the original MLE problem; see Section \ref{subsec:EM_algo}.
%more precisely, in each iteratiion first updates the Then, in each iteration, the EM algorithm first
%fixes the MLE parameters obtained in previous iteration and maximizes the objective functional with respect to the conditional distribution of the missing data, and then fixes the conditional distribution just obtained in the first step and maximizes the objective functional with respect to the MLE parameters.

Our EM-C algorithm generalizes the idea of the EM algorithm to solve multi-period finite time horizon stochastic control problems, for which there is a control policy corresponding to each time period.
The EM-C algorithm is an iterative one that updates one control policy corresponding to one time period at each step in the iterations.
%In each iteration, the EM-C algorithm updates the parameters associated with the control policies in each period one by one, with the non-updated one being unchanged,
%similar to the EM algorithm.
% then the control policy at the penultimate period, and so on.
Inheriting the spirit of the EM algorithm, the EM-C algorithm updates the control policy at a given time period by optimizing the objective function with respect to the control policy at that time period only, and with
the control policies at all other periods fixed at their most up-to-date status in the iteration of the algorithm.

What distinguishes the new EM-C algorithm from existing algorithms is fourfold: (i) Similar to the EM algorithm, the proposed EM-C algorithm has the monotonicity of performance improvement at each iteration, which leads to good convergence properties of the EM-C algorithm.  (ii) The EM-C algorithm does not assume particular
dynamics of the evolution of states (i.e. not limited to particular setting of stochastic processes), just as the EM algorithm can be applied to broad probability distributions.
(iii)  The EM-C algorithm does not use the Bellman equation; in contrast, many numerical algorithms in the literature rely on the Bellman equation or its approximation.\footnote{There are stochastic control problems for which the Bellman equation may not hold. For example, when the utility function in the general control problem \eqref{equ:multi_per_obj_gen} is not time-separable, then such problem may not have Bellman equation.} (iv) Unlike many existing algorithms, the EM-C algorithm treats finite time horizon stochastic control problems, where the optimal policy is not necessarily stationary.

%Our approach falls
%into the category of policy improvement.
%See \citet*{Powell-2011} for a comprehensive discussion on ADP.

%To summarize, we distinguish our work from existing literature in the following aspects.

%As a result, our algorithm can be suitable for finite horizon problems
%which have term structure of the optimal policy.
%Because we work in
%the policy space, our algorithm can be particularly suitable for problems
%whose policy has low dimensions but the state space has high dimensions,
%so that the optimization in each iteration can be solved relatively
%efficiently.

\subsection{Literature review}

%The EM algorithm is an iterative algorithm for finding MLE of parameters (\cite*{Dempster1977}).
%Each iteration of the algorithm  can be broken down into two steps: the expectation step (E-step) and the maximization step (M-step).
%In the E-step, the conditional expectation of the log likelihood function is estimated using the parameter estimates obtained from the previous iteration, and then in the M-step, the parameter estimates are updated through maximizing the conditional expectation of log likelihood function obtained in the E-step.
As the EM algorithm is one of the most cited algorithms in statistics, there have been numerous extensions of the algorithm; see, e.g.,
\citet*{Wei-Tanner-1990}, \citet*{Meng-Rubin-1993},
\cite*{Gu1998}, and a review in \citet[][Chap. 13]{Lange-2010}, among others. The EM algorithm allows for general distributional assumptions and has the advantageous property of monotonic convergence (\cite*{wu1983convergence}).

There is a large literature on stochastic control in economics.
\citet*{Hansen-Sargent-2013} provide detailed discussions on stochastic control problems in which
%which concerns the important class of problems with quadratic objective functions and linear transition functions,
the Bellman equations can be solved analytically. \citet*{Ljungqvist-Sargent-2013} discuss dynamic programming methods and their applications to a variety of problems in economics.
%questions in monetary policy, fiscal policy, taxation, economic growth, search theory, and labor economics.
\citet{Judd-1998} and \citet{Miranda-Fackler-2002} provide comprehensive treatment of recursive methods for solving stochastic control problems in economics.
\citet*{stokey1989recursive}  describe many examples of modeling theoretical problems in economics using dynamic programming and other recursive methods, including optimal economic growth, resource extraction, principal agent problems, public finance, business investment, asset pricing, factor supply, and industrial organization.
%
%For an extensive discussion of computational issues, see Miranda and Fackler,[11]

\citet*{Flemming-Soner-2005} provide in-depth discussion on continuous time stochastic control problems and their applications. \citet*{Kushner-Dupuis-2001} give an excellent survey of numerical methods for solving continuous time stochastic control problems by using Markov chains. There have also been many studies on the numerical solutions to continuous time stochastic control problems in mathematical finance; see, e.g., \cite*{ZHANG2004}, \cite*{Bouchard2004}, \cite*{Crisan2010}, \citet*{Fahim-Touzi-Warin-2011}, \cite*{Kharroubi2013}, \cite*{Kharroubi2013a}, and \citet*{GZZ-2014}, among others. Most of these studies focus on particular stochastic processes, e.g. discretized diffusion processes or L\'{e}vy processes, but our EM-C algorithm can be applied to general stochastic processes. Moreover, our method is a simulation based method, suitable for high dimensional problems.

Approximate dynamic programming (ADP) has been developed\footnote{ADP has also evolved under the name of reinforcement learning in computer science (see, e.g.,  \citet*{Sutton-Barto-1998}).} for dealing with the three sources of curses of dimensionality: high dimensionality of state space, control policy space, and random shock space; see the books by \citet{Powell-2011} and \citet*{Bertsekas-2012}.
%\citet*{Powell-2011} focuses on high dimensional problems with vector-valued controls and real-world applications.
ADP algorithms can be broadly classified into two categories: value iteration and policy iteration.\footnote{Many ADP algorithms focus on infinite  time horizon problems where the optimal value
function and policy are stationary. In contrast, our EM-C algorithm focuses on finite time horizon problems where neither the optimal
value function nor the optimal policy is stationary.} Most ADP algorithms are value iteration algorithms,
%{\bf Value function iteration is more popular due to its computational
%efficiency.}
which approximate the value function by employing the Bellman equation.\footnote{Value function iteration is closely related to the duality approach for stochastic dynamic programming; see \citet*{Brown-Smith-Sun-2010}, \citet*{Brown-Smith-2014}, \citet*{Brown-Haugh-2014}.}
%\citet*{Balseiro-Brown-2016}, etc.}.
Such algorithms are efficient when the value function can be well approximated, but there is no guarantee of monotonicity of value function improvement otherwise.
As an alternative, a policy iteration algorithm keeps track of
the policy instead of the value function. At each period, a value function is calculated based on a policy estimated previously and then improved within the policy space. %\citet*{Brown-Smith-Sun-2010} provides a dual approach for obtaining an upper bound on the optimal value function by relaxing and penalizing the nonanticipativity constaints;
The value iteration and policy iteration ADP algorithms may not have monotonic improvement of the value function at each iteration.

Our algorithm is related to but is fundamentally different from the policy iteration ADP algorithms mainly in that: (i) the EM-C algorithm does not use the Bellman equation; (ii) the EM-C algorithm has monotonic improvement of the value function at each iteration; and (iii) the EM-C algorithm can be applied to general control problems in which the objective functions may not be time-separable.

ADP is closely related to the problem of American option pricing using simulation. \citet*{Broadie-Glasserman-1997} develop an implicit approximate dynamic programming algorithm for American option pricing that assigns equal weights to each branch in a randomly sampled tree. \cite*{Longstaff2001} and \citet*{Tsitsiklis-VanRoy-2001} combine simulation with regression on a set of basis functions to develop low-dimensional approximation to value functions; they are related to the stochastic mesh method introduced in
\citet*{Broadie-Glasserman-2004} and correspond to an implicit choice of mesh weights. See also \citet[][Ch. 8]{Glasserman-2004} for more discussion.

The literature of Markov decision processes mainly concerns multi-period stochastic control problems with a finite state space or a finite control space. There are also simulation-based algorithms for Markov decision processes; see, e.g., the books by \citet*{chang2007simulation} and \citet*{Gosavi-2015} for comprehensive review and discussion.
The main differences between these algorithms and our EM-C algorithm are: (i) The EM-C algorithm has monotonicity in each iteration; (ii) The EM-C algorithm does not utilize Bellman equation.
%MDP mainly deal with problems with small finite state space, or small finite action space, or a small number of time periods. For problems with large uncountable state/action space or a large number of time periods, MDP algorithm may suffer from curse of dimensionality.

%, e.g., the books by \citet*{CFHM-2007b} propose an asymptotically efficient simulation-based algorithm for solving large finite horizon stochastic dynamic programming problems.

%\cite*{Kharroubi2013,Kharroubi2013a} propose general methods to solve fully nonlinear HJB type equations associated to stochastic control, and \cite*{ZHANG2004,Crisan2010,Bouchard2004} study numerical methods for discretized BSDEs.

The rest of the paper is organized as follows. In Section \ref{sec:cem_algo}, we propose the EM-C algorithm. In Section \ref{sec:Convergence-Analysis}, we show that the EM-C algorithm improves the objective function in each iteration and hence has good convergence properties. In Section \ref{sec:Simulation_Based_Algorithm}, we propose an implementation of the EM-C algorithm based on simulation and the stochastic approximation algorithm. Section \ref{sec:dynamic_pricing} and Section \ref{sec:business_cycle} present two applications of the EM-C algorithm in monopoly pricing for airline tickets and real business cycles respectively.

\section{The EM-Control (EM-C) Algorithm}\label{sec:cem_algo}

\subsection{The EM Algorithm}\label{subsec:EM_algo}

Suppose we observe the data $z$ of a random vector $Z$ but not the data of the random vector $Y$. Assume that the joint probability density function of $X=(Y, Z)$ is given by $p(y, z\mid \theta)$ with $\theta$ being the parameter. The probability density function of $Z$ is given by $p(z\mid \theta)$. The MLE of the parameter $\theta$ is obtained by maximizing the log likelihood $L(\theta)=\log p(z\mid \theta)$.

Starting from an initial estimate $\theta^0$, at the $k$th iteration the EM algorithm updates $\theta^{k-1}$ to be $\theta^k$ by two steps:
\begin{enumerate}
  \item
  E step: Compute $q^{k}(y)=p(y\mid z, \theta^{k-1})$, which
  is the conditional density of the missing data $y$ given the observed data $z$ and the parameter estimate $\theta^{k-1}$ obtained from previous iteration.
  \item
  M step: Set $\theta^{k}$ to be the $\theta$ that maximizes
  $$E_{q^k}[\log p(y, z\mid \theta)]:=\int \log p(y, z\mid \theta)q^{k}(y)dy,$$
  where $E_{q^{k}}$ denotes the expectation with respect to $y$ under the conditional distribution $q^{k}$.
\end{enumerate}

\citet*{Neal-Hinton-1999} provides an alternative view of the EM algorithm in which both the E-step and the M-step are maximizing (or at least not decreasing) the same objective functional. In fact, define a functional $F(q, \theta)$ as
\begin{equation}\label{equ:EM_obj}
  F(q, \theta):=E_{q}[\log p(y, z\mid \theta)] + H(q)=\int \log p(y, z\mid \theta)q(y)dy+ H(q),
\end{equation}
where $H(q):=-\int \log q(y)\cdot q(y)dy$ is the entropy of the probability density $q$. Then, \citet[][Theorem 1]{Neal-Hinton-1999} shows that the E-step and M-step of the EM algorithm at the $k$th iteration are equivalent to
\begin{enumerate}
  \item
  E step: Set $q^{k}$ to be $\arg\max_q F(q, \theta^{k-1})$.
  \item
  M step: Set $\theta^{k}$ to be $\arg\max_{\theta} F(q^k, \theta)$.
\end{enumerate}
Hence, at each iteration, the EM algorithm first maximizes the objective functional $F(q, \theta)$ with respect to $q$ only and with $\theta$ fixed, and then maximizes $F(q, \theta)$ with respect to $\theta$ only and with $q$ fixed.
%Hence, at each iteration, the EM algorithm iteratively maximizes the objective functional $F(q, \theta)$ with respect to $q$ and $\theta$ respectively.

The EM algorithm allows for very general distribution
assumption for $(Y, Z)$; it also has monotonicity in each iteration which lead to good convergence properties (\citet{wu1983convergence}).

%\subsection{One-period Stochastic Control Problem}
%
%Now we consider a simple 1-period stochastic control problem
%
%\begin{eqnarray}
%\max_{c} &  & E\left[u(s, c)\right]\nonumber \\
%\text{s.t.} &  & s =\psi(c,z),\label{eq:s_dyn_simple_ex_stoch_noise-1-1}
%\end{eqnarray}
%where $c$ is the control policy that affects the distribution
%of the state $s$ through a function $\psi$ and $z$ is the random source
%of this system.
%
%The problem \eqref{eq:s_dyn_simple_ex_stoch_noise-1-1} is equivalent to
%\be\label{equ:SA_obj_one_period}
%\max_{c} E\left[u(\psi(c,z), c)\right].
%\ee
%Such problem can be solved by Stochastic Approximation (SA) or other
%stochastic search methods.

\subsection{The Multi-Period Finite Time Horizon Stochastic Control Problem}

%Now we extend the problem \eqref{eq:EM-1-1} to
Now we consider a general multi-period finite time horizon stochastic control problem, which allows for vector-valued control policies, vector-valued states, and vector-valued random shocks. Let $n_c$ be the dimension of the control policy and let $n_s$ be the dimension of the state.
Suppose there are $T$ time periods and at period $0$ a decision maker observes the initial state $s_{0}\in \mathbb{R}^{n_s}$.\footnote{Without loss of generality, in this paper, we assume the initial state $s_{0}$ is known at period 0. In fact, if $s_0$ is random in a problem, one can simply relabel period 0 in that problem to be period 1 and then the random $s_0$ in that problem becomes $s_1$ in our problem formulation.}
At the $t$th period the decision maker observes the state $s_t\in \mathbb{R}^{n_s}$ and then chooses a $n_c$-dimensional control $c_t\in\sigma(s_t)$, the sigma field generated by $s_t$. Hence,
the policy $c_{t}$ is adapted to the information available up to period $t$ and can be represented as a function of $s_t$. Since $s_0$ is known at period $0$, $c_0\in\mathbb{R}^{n_c}$ is also deterministic. For $t\geq 1$, we assume that
\begin{equation}\label{eq:c_t_s_t}
  c_t = c(t, s_t, \theta_t), t\geq 1,
\end{equation}
where $c(\cdot)$ is a function and $\theta_t=(\theta_{t, 1}, \theta_{t, 2}, \ldots, \theta_{t, d})'\in \mathbb{R}^d$ is the vector of parameters for the $t$th period. For example, one may assume that the policy $c_t$ is linearly spanned by
%Let $\Phi_t$ be the policy subspace that is linearly spanned
a set of basis functions, i.e.,
%\begin{equation}\label{equ:linear_policy_c}
 $ c_t:=\sum_{i=1}^{d}\theta_{t, i}\phi_{t, i}(s_t)$,  $t\geq 1,$
%\end{equation}
where $\{\phi_{t, i}: \mathbb{R}^{n_s}\to \mathbb{R}^{n_c}, i=1, \ldots, d\}$ is the set of basis functions for the $t$th period.
%\begin{equation}\label{equ:linear_policy_space}
%  \Phi_t:=\left\{\sum_{i=1}^{d}\theta_{i, t}\phi_{i, t}(\cdot)\mid \theta_{i, t}\in\mathbb{R}, i=1,\ldots, d\right\},
%\end{equation}
%Hence, $c_t\in\Phi_t$ means that there exist $\theta_t=(\theta_{1, t}, \theta_{2, t}, \ldots, \theta_{d, t})'\in\mathbb{R}^d$ such that $c_t=\sum_{i=1}^d \theta_{i, t}\phi_{i,t}(s_t)$. %, $t=0, \ldots, T-1$.
%We consider the policy to be a function
%of current states.
The state $s_{t+1}$ is determined by $s_t$ and the control $c_t$ by the following state evolution equation
%from time $t$ to $t+1$ is described by
\begin{equation}\label{eq:state_evo}
  s_{t+1}=\psi_{t+1}(s_{t},c_{t},z_{t+1}),
\end{equation}
where $\psi_{t+1}(\cdot)$ is the state evolution function and
$z_{t+1}\in\mathbb{R}^{n_z}$ is the random vector denoting the random shock in the $(t+1)$th period.
Path dependence can be accommodated by including auxiliary
variables in $s_t$.
The state evolution dynamics in \eqref{eq:state_evo} is a general one, which is not restricted to discretized diffusion processes or L\'{e}vy
processes.

At period 0, the decision maker wishes to choose the optimal control $c_0\in\mathbb{R}^{n_c}$ and the sequence of control parameters $\theta_1, \ldots, \theta_{T-1}$, which determines the sequence of controls $c_1, \ldots, c_{T-1}$, so as to maximize the expectation of his or her utility
{\allowdisplaybreaks
\begin{align}
\max_{(c_0, \theta_{1}, \ldots, \theta_{T-1})\in \Theta}\ \ & E_0\left[\sum_{t=0}^{T-1}u_{t+1}(s_{t+1}, s_t, c_{t})\middle | c_0, \theta_{1}, \ldots, \theta_{T-1}\right]\label{equ:multi_per_obj}\\
 \text{s.t.}\,\ \ \ \ \ \ \ \  & c_t = c(t, s_t, \theta_t), t = 0, 1, \ldots, T-1,\label{equ:control}\\
 & s_{t+1}=\psi_{t+1}(s_{t},c_{t},z_{t+1}), t = 0, 1, \ldots, T-1,\notag%\label{equ:state_evol}
\end{align}}%
where $\Theta$ is a subset of $\mathbb{R}^{n}$ with $n=n_c+(T-1)d$; $u_{t+1}(\cdot)$ is the utility function of the decision maker in the $(t+1)$th period. It is worth noting that the utility function in the first period can include utility at period $0$.

A control problem more general than the problem \eqref{equ:multi_per_obj} is given by
{\allowdisplaybreaks
\begin{align}
\max_{(c_0, \theta_{1}, \ldots, \theta_{T-1})\in \Theta}\ \ & E_0\left[u(s_0, c_0, s_1, c_1, \ldots, s_{T-1}, c_{T-1}, s_T)\middle | c_0, \theta_{1}, \ldots, \theta_{T-1}\right]\label{equ:multi_per_obj_gen}\\
 \text{s.t.}\,\ \ \ \ \ \ \ \  & c_t = c(t, s_t, \theta_t), t = 0, 1, \ldots, T-1,\notag\\
 & s_{t+1}=\psi_{t+1}(s_{t},c_{t},z_{t+1}), t = 0, 1, \ldots, T-1,\notag %\label{equ:state_evol_gen}
\end{align}}%
where $u(s_0, c_0, s_1, c_1, \ldots, s_{T-1}, c_{T-1}, s_T)$ is a general utility function that may not be time-separable as the one in \eqref{equ:multi_per_obj}. For simplicity of exposition, we will present our EM-C algorithm for the problem \eqref{equ:multi_per_obj}; however, the EM-C algorithm also applies to the general problem \eqref{equ:multi_per_obj_gen}; see Appendix \ref{app:CEM-general-control} for details.

For simplicity of notation, we denote $x=(c_0, \theta_1, \theta_2, \ldots, \theta_{T-1})$ and denote
the objective function of problem \eqref{equ:multi_per_obj} by
 \be\label{equ:utility_func}
 U(x):=U(c_0, \theta_1, \theta_2, \ldots, \theta_{T-1}):=E_0\left[\sum_{t=0}^{T-1}u_{t+1}(s_{t+1}, s_t, c_t)\middle | c_0, \theta_{1}, \ldots, \theta_{T-1}\right].
 \ee
In general, the expectation in \eqref{equ:utility_func} cannot be evaluated in closed-form, and hence $U(x)$ does not have an analytical form.

%, $z_{t+1}$
%is the random source for the $(t+1)$th period.
%where $\phi_{i}$ is function of states variables.

\subsection{Description of the EM-Control (EM-C) Algorithm}%\label{sub:Algorithm}

In this subsection, we generalize the idea of the EM algorithm to propose the EM-Control (EM-C) algorithm for solving \eqref{equ:multi_per_obj}.
The EM-C algorithm is an iterative algorithm involving multiple rounds of the back-to-front updates. Inheriting the spirit of the EM algorithm, the EM-C algorithm updates the control policy at a given time period by optimizing the objective function with respect to the control policy at that time period only, and with
the control policies at all other periods fixed at their most up-to-date status in the iteration of the algorithm.

More precisely, suppose that after the $(k-1)$th iteration, the control policy parameter is $x^{k-1}:=(c^{k-1}_0, \theta^{k-1}_{1}, \theta^{k-1}_{2}, \ldots, \theta^{k-1}_{T-1})$. In the $k$th iteration, the EM-C algorithm updates $x^{k-1}$ to be $x^{k}:=(c^{k}_0, \theta^{k}_{1}, \theta^{k}_{2}, \ldots, \theta^{k}_{T-1})$ by the updating rule:
    \be\label{equ:def_pts_map}
    x^k\in M(x^{k-1}),
    \ee
 where $M(\cdot)$ is a point-to-set map on $\Theta$ (i.e., $M(\cdot)$ maps a point in $\Theta$ to a subset of $\Theta$) that represents the updating rule. The EM-C algorithm updates $c^{k-1}_0, \theta^{k-1}_{1}, \theta^{k-1}_{2},\ldots, \linebreak \theta^{k-1}_{T-1}$ backward in time; at each time period $t=T-1, T-2, \ldots, 1$, the algorithm updates $\theta^{k-1}_{t}$ to be $\theta^{k}_{t}$ and then moves backward to update $\theta^{k-1}_{t-1}$; at last, the algorithm updates $c^{k-1}_0$ to be $c^{k}_0$.

Next, we specify the precise updating rule in \eqref{equ:def_pts_map}. In the $k$th iteration, before updating the control parameter at period $t\in \{T-1, T-2, \ldots, 1\}$, the control policy parameter is $(c^{k-1}_0, \theta^{k-1}_{1}, \ldots, \theta^{k-1}_{t-1}, \theta^{k-1}_{t}, \theta^{k}_{t+1}, \theta^{k}_{t+2}, \ldots, \theta^{k}_{T-1})$. Then, at period $t$, the EM-C algorithm updates $\theta^{k-1}_{t}$ to be $\theta^{k}_{t}$ such that
{\allowdisplaybreaks
\begin{align}\label{eq:monot}
            & U(c_0^{k-1}, \theta_1^{k-1}, \theta_2^{k-1}, \ldots, \theta_{t-1}^{k-1}, \theta_{t}^{k}, \theta_{t+1}^{k}, \ldots, \theta_{T-1}^{k})\notag\\
             \geq{} & U(c_0^{k-1}, \theta_1^{k-1}, \theta_2^{k-1}, \ldots, \theta_{t-1}^{k-1}, \theta_{t}^{k-1}, \theta_{t+1}^{k}, \ldots, \theta_{T-1}^{k}),
\end{align}}%
which can be easily shown to be equivalent to
{\allowdisplaybreaks
\begin{align}\label{equ:new_2}
&  E_0\left[\sum_{j=t}^{T-1} u_{j+1}(s_{j+1},s_{j},c_j)\middle | c_0^{k-1}, \theta_1^{k-1}, \ldots, \theta_{t-1}^{k-1}, \theta_t^k, \theta_{t+1}^k, \ldots, \theta_{T-1}^k \right]\notag\\
\geq{} &  E_0\left[\sum_{j=t}^{T-1} u_{j+1}(s_{j+1},s_{j},c_j)\middle | c_0^{k-1}, \theta_1^{k-1}, \ldots, \theta_{t-1}^{k-1}, \theta_t^{k-1}, \theta_{t+1}^k, \ldots, \theta_{T-1}^k \right];
\end{align}}%
see Appendix \ref{app:simple_deriv} for details.
Therefore, such $\theta^{k}_{t}$ that satisfies \eqref{eq:monot} can be obtained by finding a suboptimal (optimal) solution to the problem
{\allowdisplaybreaks
            \begin{align*}
              %\theta_{t}^{k} = \arg
              \max_{\theta_{t}\in \Theta_t}\
                   & E_0\left[\sum_{j=t}^{T-1} u_{j+1}(s_{j+1},s_j, c_j)\middle | c_0^{k-1}, \theta_1^{k-1}, \ldots, \theta_{t-1}^{k-1}, \theta_t, \theta_{t+1}^k, \ldots, \theta_{T-1}^k\right],%\label{eq:opt_t_S}
            \end{align*}}%
where $\Theta_t=\{\theta\in \R^d\mid (c_0^{k-1}, \theta_1^{k-1}, \ldots, \theta_{t-1}^{k-1}, \theta, \theta_{t+1}^k, \ldots, \theta_{T-1}^k)\in\Theta\}$. After $\theta^{k-1}_{t}$ is updated to be $\theta^{k}_{t}$, the control policy parameter is updated from $(c^{k-1}_0, \theta^{k-1}_{1}, \ldots, \theta^{k-1}_{t-1}, \theta^{k-1}_{t},\linebreak \theta^{k}_{t+1},\ldots, \theta^{k}_{T-1})$ to $(c^{k-1}_0, \theta^{k-1}_{1}, \theta^{k-1}_{t-1}, \theta^{k}_{t}, \theta^{k}_{t+1},\ldots, \theta^{k}_{T-1})$.

Similarly, at period 0, before $c_0^{k-1}$ is updated, the control policy parameter is
$(c^{k-1}_0, \theta^{k}_{1}, \ldots, \theta^{k}_{T-1})$. Then, the EM-C algorithm updates $c_0^{k-1}$ to be $c_0^k$ such that
            \begin{align}\label{eq:mono_0}
            & U(c_0^{k}, \theta_1^{k}, \theta_2^{k}, \ldots, \theta_{T-1}^{k}) \geq U(c_0^{k-1}, \theta_1^{k}, \theta_2^{k}, \ldots, \theta_{T-1}^{k}).
            \end{align}
Then, the control parameter is updated from $(c^{k-1}_0, \theta^{k}_{1}, \ldots, \theta^{k}_{T-1})$ to $(c^{k}_0, \theta^{k}_{1}, \ldots, \theta^{k}_{T-1})$.

 In short, Algorithm \ref{alg:Main_Algorithm} summarizes the EM-C algorithm for solving problem \eqref{equ:multi_per_obj}.

\begin{algorithm}[htbp]
\caption{The EM-C algorithm for solving problem \eqref{equ:multi_per_obj}}
\label{alg:Main_Algorithm}
\begin{enumerate}
\item Initialize $k=1$ and $x^0=(c^0_0, \theta^0_{1}, \theta^0_{2}, \ldots, \theta^0_{T-1})$.  %Then,  $c^0_t=\sum_{i=1}^d\theta_{i,t}^{0}\phi_{i,t}(\cdot)$, $t=1, \ldots, T-1$.
\item Iterate $k$ until some stopping criteria are met. In the $k$th iteration, update $x^{k-1}=(c^{k-1}_0, \theta^{k-1}_{1}, \theta^{k-1}_{2}, \ldots, \theta^{k-1}_{T-1})$ to $x^{k}=(c^{k}_0, \theta^{k}_{1}, \theta^{k}_{2}, \ldots, \theta^{k}_{T-1})$
    by moving backwards from $t=T-1$ to $t=0$ as follows:
%    \begin{enumerate}
%    \item[(a)] At time $T-1$, update $\theta_{T-1}^{k-1}$ to be $\theta_{T-1}^{k}$ such that
%{\allowdisplaybreaks
%\begin{align}
%& E_{0}\left[u_{T}(s_{T}, s_{T-1}, c_{T-1})\middle | c_0^{k-1}, \theta_1^{k-1}, \ldots, \theta_{T-2}^{k-1}, \theta_{T-1}^k\right]\notag\\
%\geq & E_{0}\left[u_{T}(s_{T}, s_{T-1}, c_{T-1})\middle | c_0^{k-1}, \theta_1^{k-1}, \ldots, \theta_{T-2}^{k-1}, \theta_{T-1}^{k-1}\right].
%\end{align}}%
%Such an $\theta^{k}_{T-1}$ can be set as a suboptimal (optimal) solution to the problem
%\begin{equation}\label{eq:cem_opt_T_1_S}
%        %\theta_{T-1}^{k} = \arg
%        \max_{\theta_{T-1}\in \R^{d}}
%       E_{0}\left[u_{T}(s_{T}, s_{T-1}, c_{T-1})\middle | c_0^{k-1}, \theta_1^{k-1}, \ldots, \theta_{T-2}^{k-1}, \theta_{T-1}\right].
%\end{equation}
%            Noting that in the above expectation, the sequence of states $(s_{1}, s_{2}, \ldots, s_{T})$ originate from the initial state $s_0$ and then evolve according to the control policy parameter $(c_0^{k-1}, \theta^{k-1}_{1}, \ldots, \theta^{k-1}_{T-2}, \theta_{T-1})$.

    \item[(a)] Move backward from $t=T-1$ to $t=1$. At each period $t$, update $\theta_{t}^{k-1}$ to be $\theta_{t}^{k}$ such that
{\allowdisplaybreaks
\begin{align*}
&  E_0\left[\sum_{j=t}^{T-1} u_{j+1}(s_{j+1},s_{j},c_j)\middle | c_0^{k-1}, \theta_1^{k-1}, \ldots, \theta_{t-1}^{k-1}, \theta_t^k, \theta_{t+1}^k, \ldots, \theta_{T-1}^k \right]\notag\\
\geq{} &  E_0\left[\sum_{j=t}^{T-1} u_{j+1}(s_{j+1},s_{j},c_j)\middle | c_0^{k-1}, \theta_1^{k-1}, \ldots, \theta_{t-1}^{k-1}, \theta_t^{k-1}, \theta_{t+1}^k, \ldots, \theta_{T-1}^k \right].
\end{align*}}%
Such $\theta^{k}_{t}$ can be set as a suboptimal (optimal) solution to the problem
            \begin{align}\label{eq:opt_t_S}
              %\theta_{t}^{k} = \arg
       &    \max_{\theta_{t}\in \Theta_t}
                    E_0\left[\sum_{j=t}^{T-1} u_{j+1}(s_{j+1},s_j, c_j)\middle | c_0^{k-1}, \theta_1^{k-1}, \ldots, \theta_{t-1}^{k-1}, \theta_t, \theta_{t+1}^k, \ldots, \theta_{T-1}^k\right].
            \end{align}

%            Noting that in the above expectation, the sequence of states $(s_{1}, s_{2}, \ldots, s_{T})$ originate from the initial state $s_0$ and then evolve according to the control policy parameter $(c_0^{k-1}, \theta^{k-1}_{1}, \ldots, \theta^{k-1}_{t-1}, \theta_t, \theta_{t+1}^k, \ldots, \theta_{T-1}^k)$.
    \item[(b)] At period $0$, update $c_{0}^{k-1}$ to be $c_{0}^{k}$ such that
            \begin{align*}
            & E_0\left[\sum_{j=0}^{T-1} u_{j+1}(s_{j+1}, s_j, c_j)\middle | c_0^k, \theta_1^{k}, \ldots, \theta_{T-1}^k \right]\notag\\
            \geq{} & E_0\left[\sum_{j=0}^{T-1} u_{j+1}(s_{j+1}, s_j, c_j)\middle | c_0^{k-1}, \theta_1^{k}, \ldots, \theta_{T-1}^k \right].
            \end{align*}
Such $c^{k}_0$ can be set as a suboptimal (optimal) solution to the problem
            \begin{align}\label{eq:opt_0}
                    %&c_0^{k}
%                    \arg\max_{c_0\in \R^{n_c}}
%                   U(c_0, \theta_1^{k}, \ldots, \theta_{T-1}^k)\notag\\
                    %=\arg
                    &\max_{c_0\in \Theta_0}
                   E_0\left[\sum_{j=0}^{T-1} u_{j+1}(s_{j+1},s_j, c_j)\middle | c_0, \theta_1^{k}, \ldots, \theta_{T-1}^{k} \right],
            \end{align}
where $\Theta_0=\{c\in\mathbb{R}^{n_c}\mid (c,\theta_1^{k}, \ldots, \theta_{T-1}^{k})\in\Theta \}$.
%                Noting that the sequence of states $(s_{1}, s_{2}, \ldots, s_{T})$ in the above expectation originate from the initial state $s_0$ and then evolve according to the sequence of controls $(c_0, \theta_{1}^{k}, \ldots, \theta_{T-1}^{k})$.
%     \end{enumerate}
\end{enumerate}
\end{algorithm}

Two remarks are in order: (i) In the EM-C algorithm, when we update $\theta_t^{k-1}$ to $\theta_t^k$ or update $c_0^{k-1}$ to $c_0^k$, if no improvement of the objective function can be found, we simply set $\theta_t^k=\theta_t^{k-1}$ or set $c_0^{k}=c_0^{k-1}$.
%When we update $\theta_t^{k-1}$ to be $\theta_t^k$ or updating $c_0^{k-1}$ to $c_0^k$, we need to evaluate the expectation in %\eqref{eq:cem_opt_T_1_S},
%\eqref{eq:opt_t_S} and \eqref{eq:opt_0}, where the expectation is evaluated with all the parameters in other time periods fixed; this  corresponds to the E-step in the EM algorithm. And then, the maximization in %\eqref{eq:cem_opt_T_1_S},
%\eqref{eq:opt_t_S} and \eqref{eq:opt_0} corresponds to the M-step in the EM algorithm.
(ii) The EM-C algorithm does not use the dynamic programming principle (i.e., the Bellman equation). In contrast, the ADP algorithms in the literature are based on the Bellman equation. Furthermore,
because the EM-C algorithm does not use the Bellman equation, it
can be applied to the general control problem \eqref{equ:multi_per_obj_gen} where the utility function may not be time-separable. See Appendix \ref{app:CEM-general-control} for details.
%can be applied to stochastic control problem that do not satisfy the dynamic programming principle.

%Next, we replace M-step by dynamic programming and stochastic approximation.
%We start from the last period. With policies in all the other periods
%fixed, we optimize the value function to update the policy only in
%this period. Then we move to the previous time step and perform the
%same optimization procedure. We repeat this until we reach the first
%step. This updates the policy for one round.

%\comm{In the Algorithm \ref{alg:Main_Algorithm}, at each iteration $k$, we update the control parameters in a time-backward manner, i.e., we update $\theta_t^{k-1}$ before we update $\theta^{k-1}_{t-1}$. Theoretically, we can also update the control parameters in a time-forward manner, in which case the algorithm has the same convergence property as in the time-backward case. However, it is much more computationally efficient to update the control parameters backward than forward if the algorithm is implemented based on simulation; see Section \ref{sec:Simulation_Based_Algorithm} for more details.}

The intuition of  the view of the EM algorithm in \citet*{Neal-Hinton-1999} and that of our EM-C algorithm are also related to
the block coordinate descent (BCD) algorithms, in which the coordinates are divided into blocks and only one block of coordinates are updated at each substep of iterations in a cyclic order. However, the details of the algorithms differ significantly:
(i) In essence, the EM-C algorithm attempts to update control policies,  just like the EM algorithm that can be viewed as a generalized BCD searching in the functional space (i.e., space of distribution $q$ in \eqref{equ:EM_obj}) rather than space of real numbers. That is why the proof of convergence of EM-C algorithm is similar to that of the EM algorithm (e.g. as in \cite*{wu1983convergence}).
(ii) BCD methods are used for maximizing deterministic objective functions, but the EM-C algorithm is used for maximizing the expectation of a random utility function (i.e., \eqref{equ:utility_func}), which generally cannot be evaluated analytically. That is why we have to employ simulation and stochastic optimization to implement the EM-C algorithm (see Section \ref{sec:Simulation_Based_Algorithm}).
%BCD methods are used for solving deterministic optimization problems where each coordinate is a number, but, similar to the EM algorithm, what are updated in our EM-C algorithm are functionals (control policies) not numbers.
(iii) The EM-C algorithm is more flexible in the optimization requirement.
Unlike the BCD algorithms, the EM-C algorithm does not require to update the control parameter to be the exact minimizer of the subproblem (\eqref{eq:opt_t_S} or \eqref{eq:opt_0}), nor does it update the control parameter based on the gradient of the objective function, partly because in the problems solvable by the EM-C algorithm typically neither the objective function (i.e., \eqref{equ:utility_func}) nor the gradient of the objective function can be evaluated analytically.
(iv) The convergence of the EM-C algorithm holds under weaker conditions. The convergence of the BCD algorithms is obtained based on various assumptions on the objective function such as that the objective function is convex or is the sum of a smooth function and a convex separable function or satisfies certain separability and regularity conditions;\footnote{\citet{Luo-Tseng-1992} prove the convergence of the coordinate descent (CD) algorithm when the objective function is strictly convex twice continuously differentiable. \citet[][Chap. 2.7]{bertsekas1999nonlinear} shows the convergence of the CD algorithm when
the exact minimizer of each subproblem is unique and is used to update a block of coordinates. \citet{Tseung-2001} studies the convergence properties of a block CD method when the objective function has certain separability and regularity properties and when the exact minimizer of each subproblem is used to update a block of coordinates.
% such as that the objective function is pseudoconvex in every pair of coordinate blocks from among
\citet{Wright-2015} discusses the convergence of CD algorithms when the objective function is convex and when the coordinates are updated based on the gradient of the objective function.} in contrast, the proof of convergence of EM-C algorithm is similar to that of the EM algorithm, as in \cite*{wu1983convergence}, which does not need such assumptions on the objective function.  See Section \ref{sec:Convergence-Analysis} for details.

\section{Convergence Analysis}\label{sec:Convergence-Analysis}

The convergence properties of EM-C algorithm are similar to those of the EM algorithm.
First, the EM-C algorithm has monotonicity in each iteration. Second, under mild assumptions, the sequence of objective function values generated by the iteration of EM-C algorithm converges to a stationary value (i.e., objective function value evaluated at a stationary point) or a local maximum value. Third, the sequence of control parameters generated by the iteration of EM-C algorithm converges to a stationary point or a local maximum point under some additional regularity conditions.

\subsection{Monotonicity}

\begin{thm}\label{thm:monoto}
  The objective function $U(\cdot)$ defined in \eqref{equ:utility_func} monotonically increases in each iteration of the EM-C algorithm, i.e.,
  \be\label{equ:monoto}
  U(x^k)=U(c^k_0, \theta^k_1, \theta^k_2, \ldots, \theta^k_{T-1})\geq U(x^{k-1})=U(c^{k-1}_0, \theta^{k-1}_1, \theta^{k-1}_2, \ldots, \theta^{k-1}_{T-1}), \forall k.
  \ee
\end{thm}
\begin{proof}
See Appendix \ref{app:proof_monoto}.
\end{proof}

\subsection{Convergence of $\{U(x^k)\}_{k\geq 0}$ to a Stationary Value or a Local Maximum Value}%\label{subsec:converg_u_x}

Let $\{x^k\}_{k\geq 0}$ be the sequence of control parameters generated by the EM-C algorithm. In this subsection, we consider the issue of the convergence of $U(x^k)$ to a stationary value or a local maximum value. We make the following mild assumptions on the objective function $U(\cdot)$ defined in \eqref{equ:utility_func}:
%\begin{assumption}\label{ass:U}
{\allowdisplaybreaks
\begin{align}
  &\text{For any }x^0\ \text{such that}\ U(x^0)>-\infty, \{x\in\Theta\mid U(x)\geq U(x^0)\}\ \text{is compact.} \label{equ:assump_1}\\
  &U(\cdot)\ \text{is continuous in}\ \Theta\ \text{and differentiable in the interior of}\ \Theta.\label{equ:assump_2}
\end{align}}%
%\end{assumption}
The assumption \eqref{equ:assump_2} is needed as we need to define stationary points of $U(\cdot)$. Suppose the objective function $U(\cdot)$ satisfies \eqref{equ:assump_1} and \eqref{equ:assump_2}. Then, we have
\be\label{equ:bounded_above}
\{U(x^k)\}_{k\geq 0}\ \text{is bounded above for any }x^0\ \text{such that}\ U(x^0)>-\infty.
\ee
By \eqref{equ:monoto} and \eqref{equ:bounded_above}, $U(x^k)$ converges monotonically to some $U^*$. It is not guaranteed that $U^*$ is the global maximum of $U$ on $\Theta$. In general, if the objective function $U$ has several local maxima and stationary points, which type of points the sequence generated by the EM-C algorithm converges to depends on the choice of the starting point $x^0$; this is also true in the case of the EM algorithm.

A map $\rho$ from points of $X$ to subsets of $X$ is called a point-to-set map on $X$ (\cite*{wu1983convergence}). Let $M$ be the point-to-set map of the EM-C algorithm defined in \eqref{equ:def_pts_map}. Define
{\allowdisplaybreaks
\begin{align}
  \mathcal{M}&:=\text{set of local maxima of }U(\cdot)\ \text{in}\ \Theta,\notag\\
  \mathcal{S}&:=\text{set of stationary points of }U(\cdot)\ \text{in}\ \Theta,\notag\\
  \mathcal{M}(a)& :=\{x\in \M \mid U(x)=a\},\label{equ:M_inverse}\\
  \mathcal{S}(a)&:=\{x\in \cS \mid U(x)=a\}.\label{equ:S_inverse}
\end{align}}%
We have the following theorem on the convergence of $\{U(x^k)\}_{k\geq 0}$ for the EM-C algorithm.

\begin{thm}\label{thm:convergence_red_em}
Suppose the objective function $U$ satisfies conditions \eqref{equ:assump_1} and \eqref{equ:assump_2}. Let $\{x^k\}_{k\geq 0}$ be the sequence generated by $x^k\in M(x^{k-1})$ in the EM-C algorithm.

(1) Suppose that
\be\label{equ:cond_converg}
U(x^k)>U(x^{k-1})\ \text{for any}\ x^{k-1}\notin \mathcal{S} ({\it \text{resp.}\ x^{k-1}\notin \mathcal{M}}).
\ee
Then, all the limit points of $\{x^k\}_{k\geq 0}$ are stationary points (resp. local maxima) of $U$, and $U(x^k)$ converges monotonically to $U^*=U(x^*)$ for some $x^*\in \mathcal{S}$ (resp. $x^*\in\mathcal{M}$).

(2) Suppose that at each iteration $k$ in the EM-C algorithm and for all $t$,
%$\theta_t^{k-1}$ and $c_0^{k-1}$ are updated to be
$\theta_t^{k}$ and $c_0^{k}$ are the optimal solutions to the problems %\eqref{eq:cem_opt_T_1_S},
\eqref{eq:opt_t_S} and \eqref{eq:opt_0} respectively. Then, all the limit points of $\{x^k\}$ are stationary points of $U$ and $U(x^k)$ converges monotonically to $U^*=U(x^*)$ for some $x^*\in \mathcal{S}$.
\end{thm}
\begin{proof}
  See Appendix \ref{app:proof_red_em}.
\end{proof}

\subsection{Convergence of $\{x^k\}_{k\geq 0}$ to a Stationary Point or a Local Maximum Point}
Let $\M(a)$ and $\cS(a)$ be defined in \eqref{equ:M_inverse} and \eqref{equ:S_inverse}  respectively. Under the conditions of Theorem \ref{thm:convergence_red_em}, $U(x^k)\to U^*$ and all the limit points of $\{x^k\}$ are in $\cS(U^*)$ (resp. $\M(U^*)$). However, this does not automatically imply the convergence of $\{x^k\}_{k\geq 0}$ to a point $x^*$. However, if $\cS(U^*)$ (resp. $\M(U^*)$) consists of a single point $x^*$, i.e., there cannot be two different stationary points (resp. local maxima) with the same $U^*$, then the following theorem says that $x^k\to x^*$. The following theorem also provides another condition under which $x^k\to x^*$.

\begin{thm}\label{thm:converg_x}
Let $\{x^k\}_{k\geq 0}$ be an instance of an EM-C algorithm satisfying the conditions of Theorem \ref{thm:convergence_red_em}, and let $U^*$ be the limit of $\{U(x^k)\}_{k\geq 0}$.

(1) If $\cS(U^*)=\{x^*\}$ (resp. $\M(U^*)=\{x^*\}$), then $x^k\to x^*$ as $k\to \infty$.

(2) If $\|x^{k+1}-x^k\|\to 0$ as $k\to\infty$, then, all the limit points of $x^k$ are in a connected and compact subset of $\cS(U^*)$ (resp. $\M(U^*)$).
In particular, if $\cS(U^*)$ (resp. $\M(U^*)$) is discrete, i.e., its only connected components are singletons, then $x^k$ converges to some $x^*$ in $\cS(U^*)$ (resp. $\M(U^*)$).
\end{thm}
\begin{proof}
  See Appendix \ref{app:proof_cong_x}.
\end{proof}

Of course, from a practical viewpoint, very often the convergence of value function $\{U(x^k)\}_{k\geq 0}$ to a stationary value or a local maximum value is more important than the convergence of $\{x^k\}_{k\geq 0}$.

\section{An Implementation of the EM-C Algorithm}\label{sec:Simulation_Based_Algorithm}

%Now we consider a simple 1-step stochastic control problem
%
%\begin{eqnarray}
%\max_{c\in\Upsilon} &  & E\left[u(s, c)\right]\nonumber \\
%\text{s.t.} &  & s=\psi(c,z).\label{eq:s_dyn_simple_ex_stoch_noise-1-1}
%\end{eqnarray}
%Consider that there is a control policy $c$, which affects the distribution
%of the state $s$ through a function $\psi$, $z$ is the random source
%of this system, and $\Upsilon$ is the policy space. So this problem is
%equivalent to problem (\ref{eq:EM-1-1}) when $\Upsilon$ equals to $\Theta$.
%
%Next, we introduce an iterative simulation-based approach.
%
%At the $k^{th}$ iteration, we use the policy derived from previous iteration
%$c^{k-1}$ and simulate sample paths $\{s_{l}(c^{k-1}):=\psi(c^{k-1},z_{l})\mid l=1, \ldots, N\}$
%where $N$ is the number of paths. Using these paths, we can estimate
%the expectation
%\[
%E_{c^{k-1}}\left[u(s, c)\right]:=E\left[u(s(c^{k-1}, z), c)\right]\approx\frac{1}{N}\sum_{l=1}^{N}u(s_l(c^{k-1}),c)
%\]
% where $E_{c^{k-1}}$ is an expectation
%conditional on the previous policy $c^{k-1}$ and $s(c^{k-1}, z):=\psi(c^{k-1}, z)$. This is the E-step
%of EM Algorithm.
%
%Next, the M-step becomes
%\[
%c^{k}\approx\arg\max_{c\in\Upsilon}E_{c^{k-1}}\left[u(s(c^{k-1}, z), c)\right].
%\]

\subsection{Implementing the EM-C Algorithm by Simulation}

In the EM-C algorithm, we need to find a suboptimal (optimal) solution to the problems %\eqref{eq:cem_opt_T_1_S},
\eqref{eq:opt_t_S} and \eqref{eq:opt_0}.
In practice, the expectation in the objective functions of these problems may not be evaluated in closed-form, which makes solving these problems difficult.
We propose to solve these problems by using a simulation based approach, called stochastic approximation (SA) algorithm.

% which has been studied
%in \cite*{Gu1998} and many other literature. Other stochastic optimization methods can also be
%used for solving the problems %\eqref{eq:cem_opt_T_1_S},
%\eqref{eq:opt_t_S} and \eqref{eq:opt_0}.
%
%We propose an implementation of the EM-C algorithm based on Monte Carlo simulation and the stochastic approximation (SA) algorithm, which is used for solving
%the subproblems of updating the control policy at a given time period.

The SA is a classical iterative stochastic optimization
algorithm that tries to find zeros or extrema of expectations which
cannot be computed directly.\footnote{The SA algorithm is initiated in \cite*{Robbins1951} and \citet*{Kiefer-Wolfowitz-1952}.
It has been widely used in reinforcement learning to improve policies in temporal difference methods (see, e.g., \cite*{chang2007simulation}). There is a voluminous literature on SA algorithms; see, e.g, \cite*{Gu1998},  a survey paper by \citet*{Lai-2003} and the books by \citet*{Kushner-Yin-2003} and \citet*{Spall-2003}. \citet*{broadie2011general} propose a SA algorithm that improves the finite time performance of the Kiefer-Wolfowitz algorithm.}
 More precisely, at each iteration of the EM-C algorithm, sample paths are
simulated using the current policy, and then the SA is applied to find updates of the control policy at each time period to improve the objective function.\footnote{However, it is not necessary to use the SA algorithm to implement our EM-C algorithm. One can also use other stochastic optimization algorithms such as the
cross-entropy algorithm (\citet*{rubinstein2004cross}) in the implementation.}

At the beginning of the $k$th iteration, we first simulate $N$ i.i.d. sample paths of the states $(s_0, s_1, \ldots, s_{T-1})$ according to the control parameter $(c^{k-1}_0, \theta_1^{k-1}, \ldots, \theta_{T-1}^{k-1})$, which are obtained at the end of the $(k-1)$th iteration. We denote these sample paths as
\be
(s_0, s_{1,l}^{k}, s_{2,l}^{k}, \ldots, s_{T-1,l}^{k}), l=1, \ldots, N.\notag
\ee
%In step 2(a), we apply the SA algorithm to solve the problem \eqref{eq:cem_opt_T_1_S}. The objective function of \eqref{eq:cem_opt_T_1_S} is equal to
%{\allowdisplaybreaks
%\begin{align}\label{equ:simu_E_T_1}
%&E_{0}\left[u_{T}(s_{T}, s_{T-1}, c_{T-1})\middle | c_0^{k-1}, \theta_1^{k-1}, \ldots, \theta_{T-2}^{k-1}, \theta_{T-1}\right]\notag\\
%%=& E_{0}\left[E_{T-1}\left[u_{T}(s_{T}, s_{T-1}, c_{T-1})\mid s_{T-1}\right]\right]\\
%={} & E_0\left[\frac{1}{N}\sum_{l=1}^N u_{T}(s^k_{T,l}(\theta_{T-1}), s^k_{T-1,l}, c^k_{T-1,l}(\theta_{T-1}))\right],
%\end{align}}%
%where $c^k_{T-1,l}(\theta_{T-1})=c(T-1, s^k_{T-1,l}, \theta_{T-1})$ (see \eqref{equ:control}) and $s^k_{T,l}(\theta_{T-1})$ is the state at time $T$ that is simulated starting from $s_{T-1, l}^k$ and following the control parameter $\theta_{T-1}$, i.e., $s^k_{T,l}(\theta_{T-1})=\psi_{T}(s_{T-1, l}^k, c^k_{T-1,l}(\theta_{T-1}), z_{T, l})$ (see \eqref{equ:state_evol}). The SA algorithm uses
%\be\label{equ:appro_E_T_1}
%\tilde f(\theta_{T-1}):=\frac{1}{N}\sum_{l=1}^N u_{T}(s^k_{T,l}(\theta_{T-1}), s^k_{T-1,l}, c^k_{T-1,l}(\theta_{T-1}))
%\ee
%as an approximation to the objective function when solving the problem. \comm{Hence, at each iteration of the SA algorithm (at the parameter $\theta_{T-1}$ corresponding to that iteration), we only need to simulate $N$ samples of the state at time $T$, i.e., $s^k_{T,l}(\theta_{T-1})$, $l=1, \ldots, N$. The samples $s^k_{T-1,l}$, $l=1, \ldots, N$ are the same for all iterations of the SA algorithm.}

In step 2(a) of Algorithm \ref{alg:Main_Algorithm}, we apply the SA algorithm to solve the problem \eqref{eq:opt_t_S}. The expectation in the objective function of \eqref{eq:opt_t_S} is equal to
{\allowdisplaybreaks
\begin{align}\label{equ:simu_E_t}
&E_0\left[\sum_{j=t}^{T-1} u_{j+1}(s_{j+1}, s_j, c_j)\middle | c_0^{k-1}, \theta_1^{k-1}, \ldots, \theta_{t-1}^{k-1}, \theta_t, \theta_{t+1}^k, \ldots, \theta_{T-1}^k\right]\notag\\
={} & E_0\left\{\frac{1}{N}\sum_{l=1}^N \left[u_{t+1}(s^k_{t+1,l}(\theta_t), s^k_{t,l}, c^k_{t,l}(\theta_t))\phantom{\sum_{j=t+1}^{T-1}}\right.\right.\notag\\
&\quad \quad \quad \quad \quad \quad  + \left.\left.\sum_{j=t+1}^{T-1} u_{j+1}(s^k_{j+1,l}(\theta_t), s^k_{j,l}(\theta_t), c^k_{j,l}(\theta_t))\right]\right\},
\end{align}}%
where $c^k_{t,l}(\theta_t)=c(t, s^k_{t,l}, \theta_{t})$ (see \eqref{equ:control}) and
$$(s^k_{t+1,l}(\theta_{t}), c^k_{t+1,l}(\theta_t), \ldots, s^k_{T-1,l}(\theta_{t}), c^k_{T-1,l}(\theta_t), s^k_{T,l}(\theta_{t}))$$
is a simulated sample path that starts from $s^k_{t, l}$ and then follows the control parameter $\theta_{t}, \theta_{t+1}^k, \ldots, \theta_{T-1}^k$. The SA algorithm uses
{\allowdisplaybreaks
\begin{align}\label{equ:appro_E_t}
\tilde f(\theta_{t}):=& \frac{1}{N}\sum_{l=1}^N \left[u_{t+1}(s^k_{t+1,l}(\theta_t), s^k_{t,l}, c^k_{t,l}(\theta_t))\phantom{\sum_{j=t+1}^{T-1}}\right.\notag\\
&\quad \quad \quad \quad \quad \quad  + \left.\sum_{j=t+1}^{T-1} u_{j+1}(s^k_{j+1,l}(\theta_t), s^k_{j,l}(\theta_t), c^k_{j,l}(\theta_t))\right]
\end{align}}%
as an approximation to the objective function when solving the problem. Hence, at each iteration of the SA algorithm (at the parameter $\theta_{t}$ corresponding to that iteration), we only need to simulate $N$ sample paths of the states during period $t+1$ to period $T$, i.e., $(s^k_{t+1,l}(\theta_{t}), \ldots, s^k_{T-1,l}(\theta_{t}), s^k_{T,l}(\theta_{t}))$, $l=1, \ldots, N$.
%, starting from $s^k_{t,l}$, $l=1, \ldots, N$, which have been simulated at the beginning of the $k$th iteration.
The samples $s^k_{t,l}$, $l=1, \ldots, N$ are the same for all iterations of the SA algorithm.

%$(s^k_{T,l}(\theta_{T-1}), s^k_{T-1,l}, c^k_{T-1,l})$ are i.i.d. copies of $(s_{T}, s_{T-1}, c_{T-1})$ that
%are simulated starting from $s_0$ and then following the control parameters $(c_0^{k-1}, \theta_1^{k-1}, \ldots, \theta_{T-2}^{k-1}, \theta_{T-1})$.
%\begin{align}\label{equ:tower_E_t}
%&E_0\left[u_{t+1}(s_{t+1}, s_t, c_t) +\sum_{j=t+1}^{T-1} u_{j+1}(s_{j+1}, s_j, c_j^{k})\right]= E_0\left[E_{t}\left[u_{t+1}(s_{t+1}, s_t, c_t) +\sum_{j=t+1}^{T-1} u_{j+1}(s_{j+1}, s_j, c_j^{k})\mid s_{t}\right]\right].
%\end{align}
%Therefore, the expectation in \eqref{eq:opt_t_S} can be approximated by

In step 2(b)  of Algorithm \ref{alg:Main_Algorithm}, we apply the SA algorithm to solve the problem \eqref{eq:opt_0}. The expectation in \eqref{eq:opt_0} is equal to
{\allowdisplaybreaks
            \begin{align}\label{eq:simu_E_0}
                  & E_0\left[\sum_{j=0}^{T-1} u_{j+1}(s_{j+1}, s_j, c_j)\middle | c_0, \theta_1^{k}, \ldots, \theta_{T-1}^k \right]\notag\\
                  ={} & E_0\left\{\frac{1}{N}\sum_{l=1}^N \left[u_{1}(s^k_{1,l}(c_0), s_0, c_0) +\sum_{j=1}^{T-1} u_{j+1}(s^k_{j+1,l}(c_0), s^k_{j,l}(c_0), c^k_{j,l}(c_0))\right]\right\},
            \end{align}}%
where $\{(s^k_{1,l}(c_0), c^k_{1,l}(c_0), \ldots, s^k_{T-1,l}(c_0), c^k_{T-1,l}(c_0), s^k_{T,l}(c_0))\}_{l=1}^N$ are $N$ i.i.d. sample paths of $(s_{1}, c_{1}, \ldots, s_{T-1}, c_{T-1}, s_T)$ that are simulated starting from $s_0$ and then following the control parameters $(c_0, \theta_1^{k}, \ldots, \theta_{T-1}^k)$. The SA algorithm uses
\begin{align}\label{equ:appro_E_0}
\tilde f(c_0):=\frac{1}{N}\sum_{l=1}^N \left[u_{1}(s^k_{1,l}(c_0), s_0, c_0) +\sum_{j=1}^{T-1} u_{j+1}(s^k_{j+1,l}(c_0), s^k_{j,l}(c_0), c^k_{j,l}(c_0))\right]
\end{align}
as an approximation to the objective function when solving the problem \eqref{eq:opt_0}.
The details of the SA algorithm for solving the problems %\eqref{eq:cem_opt_T_1_S},
\eqref{eq:opt_t_S} and \eqref{eq:opt_0} are described in Appendix \ref{sec:Stochastic-Approximation}.

Suppose that we use a fixed number of $m$ iterations in the SA algorithm. Then, the computational cost of solving the problems %\eqref{eq:cem_opt_T_1_S},
\eqref{eq:opt_t_S} and \eqref{eq:opt_0} are respectively %$O(m)$,
$O(m(T-t+1))$ and $O(mT)$. Hence, the computational cost of each iteration of the EM-C algorithm is $O(mT^2)$.
%then the computational cost for each iteration is proportional to
%length of remaining horizon since we need to simulate sample paths
%from current period until the end. As a result, the overall computational cost
%is $O(T^{2})$.
%It limits the application of our algorithm to
%those problems with reasonably short horizons. However, it is usually
%not a big concern. Because when the horizon is long, it becomes close
%to the infinite horizon case and a different set of techniques are
%available.

\subsection{Numerical Example: A Simple Stochastic Growth Model}%\label{sec:Simple-Stochastic-Growth}

%In this section, we will illustrate our algorithm by a simple example that has analytical solution.
%Due to the simple structure of this problem, our algorithm
%is able to attain controls that are very close to the optimal ones by
%only one round. As a result, the results in Section \ref{sec:Simple-Stochastic-Growth}
%are based on the first round.

%In the second example, which is modified from an example in \cite*{bartlett2001}, we will show that
%value function improvement does not necessarily improve the overall
%performance.

We consider a simple stochastic growth problem as follows
{\allowdisplaybreaks
\begin{align}
\max_{c_t}\ \  & E_0\left[\sum_{t=0}^{2}u_{t+1}(s_{t+1}, s_t, c_{t})\right]=E_0\left[\sum_{t=0}^{2}\log c_{t}+\log s_3 \right]\label{equ:simple_stoc_grow} \\
\text{s.t.}\ \  & s_{t+1} = \left(s_{t}-\frac{s_t}{1+\exp(c_{t})}\right)\exp(a+bz_{t+1}), t=0, 1, 2,\nonumber\\
  & s_{0}=1,\nonumber \\
  & c_{t}\in \R, t=0, 1, 2,\nonumber
\end{align}}%
where $a$ is a constant, $b>0$ is the volatility, and $z_{t+1}$, $t=0, 1, 2$,
are i.i.d. random noises with the standard normal distribution. At the $t$th time period, the amount $\frac{s_t}{1+\exp(c_{t})}$ is consumed from capital $s_{t}$,
and the remaining capital grows by a multiplication factor $\exp(a+bz_{t+1})$.
%The terminal
%condition $c_{4}=s_{4}$ indicates that
All available capital will be consumed in the end (at period $t=3$).

The problem can be solve analytically with the following optimal controls
and optimal value functions
{\allowdisplaybreaks
\begin{align}
& c_{t}^{*} = \log(3-t),\; t=0, 1, 2,\label{equ:opt_policy}\\
& V_{0}(s_{0}) =6a-4\log4+4\log s_{0}.\label{eq:opt_simple_stoch_growth}
%& V_{1}(s_{1}) =3a-3\log3+3\log s_{1},\notag\\
%& V_{2}(s_{2}) =a-2\log2+2\log s_{2}.\notag
\end{align}}%
%In order to make the objective function bounded, we need to set a
%more restrictive bound for $c_{t}$. Due to the randomness, we also
%need to set a positive lower bound for $s_{t}$. That is equivalent
%to bounding $z_{t+1}$ from below with some positive value.
%\begin{prop}
%Consider $c_{t}\in[M_{1}s_{t},(1-M_{2})s_{t}]$ (except for $c_{4}\in(0,s_{4}]$
%), and $s_{t}\in[M_{3},M_{4}]$ (or equivalently $z_{t+1}\in[-M_{3}',M_{4}']$)
%for some constants $M_{i}\in(0,\infty)$ (or $M_{1},M_{2},M_{3}',M_{4}'\in(0,\infty)$)
%for $i=1, \ldots, 4$, such that $M_{1}+M_{2}<1$, then Assumption \ref{assump:1}
%is satisfied.\end{prop}
%\begin{proof}
%For any $s_{t}$, feasible set is $\Gamma(s_{t})=[M_{1},(1-M_{2})s_{t}]$,
%so corresponding objective function is bounded. The bound $(1-M_{2})s_{t}$
%and $s_{t}\in[M_{3},M_{4}]$ (or $z_{t+1}\in[-M_{3}',M_{4}']$) guarantee
%that in the next period, the objective is still bounded. Assumption
%\ref{assump:1} is satisfied.
%\end{proof}
%In simulation, we can truncate the distribution of $z_{t+1}$. If
%we can choose $M_{3}'$ and $M_{4}'$ very large and $M_{1}$ and
%$M_{2}$ very small, the truncated distribution is almost identical
%to the original one. For example, for $M_{3}'=M_{4}'=8$, the corresponding
%truncated density function will only be perturbed by a factor of roughly
%$10^{-15}$.

To test our algorithm numerically, we choose $a=-0.1$ and $b=0.2$.
We use $N=10,000$ sample paths in the simulation and $m=2,000$ iterations in the
SA algorithm. We consider two specifications of basis functions. In the first specification, we use only one basis function
{\allowdisplaybreaks
\begin{align*}
&\phi_{1}(s)=s; c_t =\theta_{t, 1}\phi_1(s_t).
\end{align*}}%
In the second specification, we use two basis functions
{\allowdisplaybreaks
\begin{align*}
& \phi_{1}(s)=1, \phi_{2}(s)=s; c_t =\theta_{t,1}\phi_1(s_t)+\theta_{t,2}\phi_2(s_t).
\end{align*}}%
It follows from \eqref{equ:opt_policy} that the theoretical optimal policy $c_t^*$ lies in the space linearly spanned by the basis in the second specification (corresponding to optimal control parameters $\theta_t^*=(\log(3-t), 0)'$) but not in the first one. In the EM-C algorithm, we choose initial values of $c_0$ and $\theta_t$ to be $c^0_0=0, \theta^0_t = 0, \forall t.$

Figure \ref{fig:stoch_growth_obj_iter}
shows the objective function values of the EM-C algorithm over 5 iterations for the problem \eqref{equ:simple_stoc_grow} by using two specifications of basis functions. In both specifications, the EM-C algorithm converges quickly to a value close to the theoretical optimal objective function value given by \eqref{eq:opt_simple_stoch_growth} after 2 iterations, even in the first specification when only one basis function is used. Each iteration takes around 3 minutes.

%The theoretical optimal objective function value is -6.1452. The optimal objective function values obtained by the EM-C algorithm is -6.1421 (7.4659e-03) when only one basis function is used and is -6.1358 (7.4755e-03) when two basis functions are used. The numbers in the parenthesis denote standard errors of the estimate of the objective function using $N$ sample paths, which is equal to the sample standard deviation of the $N$ samples on the right-hand side of \eqref{equ:appro_E_0} divided by $\sqrt{N}$.

%This shows that even with imperfect basis functions, our algorithm
%can still achieve convergence results. In the example we study, the
%converged performance is not materially different from the true optimal
%value.

\begin{figure}[htbp]
\centering
\includegraphics[width=\textwidth]{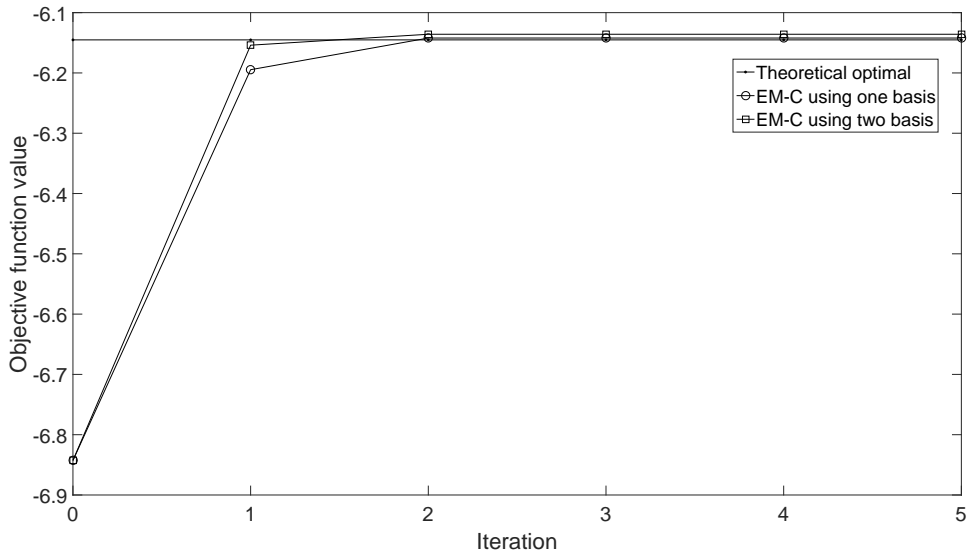}
\caption{The objective function values of the EM-C algorithm over 5 iterations for the problem \eqref{equ:simple_stoc_grow}. In the implementation, we use $N=10,000$ sample paths in the simulation and use $m=2,000$ iterations in the SA algorithm. The EM-C algorithm converges after 2 iterations. Each iteration takes around 3 minutes. The theoretical optimal objective function value is -6.1452. The optimal objective function values obtained by the EM-C algorithm is -6.1421 (7.4659e-03) when only one basis function is used and is -6.1358 (7.4755e-03) when two basis functions are used. The numbers in the parenthesis denote standard errors of the estimate of the objective function using $N$ sample paths, which is equal to the sample standard deviation of the $N$ samples on the right-hand side of \eqref{equ:appro_E_0} divided by $\sqrt{N}$.} \label{fig:stoch_growth_obj_iter}
\end{figure}

\section{Application 1: Monopoly Pricing of Perishable Products}\label{sec:dynamic_pricing}
In this section we shall apply the EM-C algorithm to solve two problems related to monopoly pricing of airline tickets. The first one, the single product airline ticket pricing, is more for the purpose of illustrating the validity of the algorithm, as there is an analytical solution available for the continuous time version of the problem and a good heuristic plug-in method for the discrete version of the problem. The second one, the multi-product airline ticket pricing, is challenging, as so far only heuristic methods are available. The EM-C algorithm not only provides a rigorous solution, but also yields significant value function improvement over the heuristic methods.

\subsection{Single Product Case}%\label{sub:inv_single}

\subsubsection{The Single Product Monopoly Pricing Model}
Consider a single product monopoly pricing for airline tickets as in \cite*{gallego1994optimal}. It is a finite horizon
problem with one state and one control.
Suppose revenue within a short period $(t, t+\Delta t)$ is given by
$p(\lambda_t)\Delta N^{\lambda}$,
where $\lambda_t$ is the sale intensity at time $t$, $N^{\lambda}$ is a Poisson counting process with intensity $\lambda_t$, $p(\lambda_t)$ is the price at time $t$, and $\Delta N^{\lambda}$
is the number of arriving customers in the time interval $(t, t+\Delta t)$. The continuous-time problem is
formulated as follows
{\allowdisplaybreaks
\begin{align}
V(n^c, T) = \sup_{\lambda_s}\ \ & E_0\left[\int_{0}^{T}p(\lambda_{s})dN_{s}^{\lambda}\right]\label{equ:prob_dyna_pricing_single_continuous}\\
\text{s.t.}\ \ & N_{T}^{\lambda}\leq n^{c},\nonumber\\
& p(\lambda_{s})=-\frac{1}{\alpha}\log\frac{\lambda_{s}}{a},\mbox{ for }s\leq T,\nonumber
\end{align}}%
where $n^{c}$ is the total remaining capacity and $T$ is the time-to-maturity.
%The price is assumed to follow a parametric function depending on the sale intensity $\lambda$.

In this problem, the state variable is the residual capacity $R{}_{s}=n^{c}-N_{s}^{\lambda}$
and the control is $\lambda_{s}$, which determines the ticket price
$p(\lambda_{s})$ and the dynamics of future arrivals.
%The capacity
%constraint is automatically taken care of because of the continuous
%setting.
Apparently, $V(n^c, 0)=V(0, T)=0$, for any $n^c$ and any $T$. When $\alpha=1$, luckily enough there is an analytical solution given by (\cite*{gallego1994optimal})
{\allowdisplaybreaks
\begin{align}
& V(n^c, t) = \log\left(\sum_{k=0}^{n^c}\frac{(aT/e)^{k}}{k!}\right),\ \text{for any}\ n^c\in\mathbb{N}^+, t>0,\label{equ:opt_revenue}\\
%\lambda^{*} & = & \frac{a}{e}\nonumber \\
%p^{*}(n^c, t) & = V(n^c, t) - V(n^c-1, t)+1,\ \text{for any}\ n^c\in\mathbb{N}^+, t>0.\label{eq:inv_true_opt_c}
& p^{*}_t = p(\lambda^*_t)  = V(R_t, T-t) - V(R_t - 1,T-t)+1,\ \text{for}\ R_t\geq 1, 0\leq t \leq T.\label{eq:inv_true_opt_c}
\end{align}}%

We discretize the time horizon $[0, T]$ into $n_{T}$ equal periods,
denoted as $t_{0}=0, \ldots, t_{n_{T}}=T$, and
formulate a discrete version of the problem \eqref{equ:prob_dyna_pricing_single_continuous} as follows:
{\allowdisplaybreaks
\begin{align}\label{equ:prob_dyna_pricing_single_prod}
\max_{c_{t_i},i=0, 1, \ldots, n_{T}-1}\ \ & E_0\left[\sum_{i=0}^{n_{T}-1}p(\lambda_{t_{i}})(N_{t_{i+1}}^{c}-N_{t_{i}}^{c})\right]\\
\text{s.t.}\ \ \ \ \ \ \ \ & N_{t_{i+1}}^{\lambda}-N_{t_{i}}^{\lambda}\overset{d}{\sim}\mbox{Poisson}(\lambda_{t_{i}}T/n_T), i=0, 1, \ldots, n_{T}-1,\label{equ:pois_dist}\\
& N_{t_{i+1}}^{c}-N_{t_{i}}^{c}=\min(n^{c}-N_{t_{i}}^{c},N_{t_{i+1}}^{\lambda}-N_{t_{i}}^{\lambda}), i=0, 1, \ldots, n_{T}-1,\label{equ:N_c}\\
  & p(\lambda_{t_{i}})=-\frac{1}{\alpha}\log\frac{\lambda_{t_{i}}}{a}, i=0, 1, \ldots, n_{T}-1,\nonumber\\
  & \lambda_{t_{i}} = \frac{a}{1+\exp(c_{t_i})}, i=0, 1, \ldots, n_{T}-1,\label{equ:logit}\\
  & c_{t_{i}} \in\mathbb{R}, i=0, 1, \ldots, n_{T}-1,\nonumber
\end{align}}%
where \eqref{equ:pois_dist} means that $N_{t_{i+1}}^{\lambda}-N_{t_{i}}^{\lambda}$ has a Poisson distribution with mean $\lambda_{t_{i}}T/n_T$;
%Unlike the continuous case, the capacity constraint
%here is omitted due to the fact that it is infeasible to control $\lambda$
%to guarantee this. There is always a positive probability of excessive
%demand ($N_{t_{i}}^{\lambda}>n^{c}$).
$N_{t}^{c}$ is the total number of customers that have arrived and bought the ticket during $[0, t]$; \eqref{equ:N_c} means that $N_{t}^{c}$ is capped at $n^c$;
\eqref{equ:logit} is used to incorporate the constraint $\lambda_{t_i}\in (0, a)$.
 In the discrete problem \eqref{equ:prob_dyna_pricing_single_prod}, the state variable is the residual capacity $R{}_{t_{i}}=n^{c}-N_{t_{i}}^{c}$.

There is no analytical solution to the discrete problem \eqref{equ:prob_dyna_pricing_single_prod}; but when $\alpha=1$, the optimal policy \eqref{eq:inv_true_opt_c} for the continuous problem can be used as a plug-in policy for the discrete problem.

\subsubsection{Numerical Results}

In the following numerical examples of problem \eqref{equ:prob_dyna_pricing_single_prod}, we choose $a=20$,
$\alpha=1$, $T=1$, $n_T=4$, and $n^{c}=20, 10, \text{and}\ 5$, respectively.
%The single-step optimal intensity $\lambda^{*}=7.358$,
%and $p^{*}=1$ (see \cite*{gallego1994optimal}).
We use $N=10,000$
sample paths in the simulation and use $m=1,000$ iteration in the SA algorithm. We specifies the control $c_t$ as the linear combination of three basis functions:
$$\phi_{i}(R):=R^{i}, i = 0, 1, 2;\ c_t =\theta_{t,1} \phi_1(R_t) + \theta_{t,2} \phi_2(R_t) + \theta_{t,3} \phi_3(R_t).$$
In the algorithm, we choose initial values of $c_0$ and $\theta_t$ to be $c^0_0=0, \theta^0_t = 0$, for all $t$.

\begin{table}[htbp]
\begin{centering}
\resizebox{\linewidth}{!}{
\begin{tabular}{cccccccccc}
\hline
\hline
\addlinespace
 & \multicolumn{3}{c}{$n^{c}=20$} & \multicolumn{3}{c}{$n^{c}=10$} & \multicolumn{3}{c}{$n^{c}=5$}\tabularnewline
 \addlinespace
 \cline{2-4}\cline{5-7}\cline{8-10}
 \addlinespace
%\hline
 & continuous & \multicolumn{2}{c}{discrete}  & continuous & \multicolumn{2}{c}{discrete}  & continuous & \multicolumn{2}{c}{discrete}\tabularnewline
 \addlinespace
 \cline{3-4}\cline{6-7}\cline{9-10}
 \addlinespace
%\hline
 & & plug-in & EM-C  & & plug-in & EM-C  & & plug-in & EM-C\tabularnewline
%\hline
mean  & 7.3576 & 7.3494 & 7.3777  & 7.2231 & 7.2207  & 7.2237  & 6.000 & 5.8964 & 5.9419\tabularnewline
%\hline
std. error & N/A & 0.0271 & 0.0270 & N/A & 0.0257  & 0.0260 & N/A & 0.0205  &  0.0204\tabularnewline
%\hline
%Skewness & 0.369 & 0.374 & 0.366 & 0.151 & 0.085  & 0.104 & -0.442 & -0.440 & -0.446\tabularnewline
%\hline
%Kurtosis & 3.123 & 3.145 & 3.119 & 2.601 & 2.460  & 2.478 & 2.464 & 2.454 & 2.437\tabularnewline
\addlinespace
\hline
\hline
\end{tabular}}
\end{centering}

\caption{Monopoly pricing of a single product: expected revenue for the
continuous problem \eqref{equ:prob_dyna_pricing_single_continuous} and the discrete problem \eqref{equ:prob_dyna_pricing_single_prod}
obtained under three policies respectively: (i) ``continuous'' means the expected revenue for the continuous problem under the theoretical optimal policy \eqref{eq:inv_true_opt_c}; (ii) ``plug-in'' means the expected revenue for the discrete problem obtained under
the plug-in policy \eqref{eq:inv_true_opt_c}; (iii) ``EM-C'' means the expected revenue for the discrete problem obtained under
the optimal policy calculated by the EM-C algorithm. The expected revenue under the theoretical optimal policy for the continuous problem is computed from \eqref{equ:opt_revenue}; the expected revenues for the discrete problem under the plug-in and EM-C policies are estimated from $N=10,000$ sample paths. We consider three cases: $n^{c} = 20, 10$, and $5$. ``Std. error''
indicates the standard error of the estimate of the expected revenue, which is equal to the sample standard deviation of the $N$ samples on the right-hand side of \eqref{equ:appro_E_0} divided by $\sqrt{N}$. \label{tab:inv_mean_rev}}
\end{table}

Table \ref{tab:inv_mean_rev} compares the
expected revenue for the
continuous problem \eqref{equ:prob_dyna_pricing_single_continuous} and the discrete problem \eqref{equ:prob_dyna_pricing_single_prod}
obtained under three policies respectively: (i) the expected revenue for the continuous problem under the theoretical optimal policy \eqref{eq:inv_true_opt_c}; (ii) the expected revenue for the discrete problem obtained under
the plug-in policy \eqref{eq:inv_true_opt_c}; (iii) the expected revenue for the discrete problem obtained under
the optimal policy calculated by the EM-C algorithm. %The expected revenue under the theoretical optimal policy for the continuous problem is computed from \eqref{equ:opt_revenue};
%the expected revenue under the other two optimal policies are estimated from $N=10000$ sample paths. We consider three cases: $n^{c} = 5, 10,\, \text{and}\, 20$.
It seems that the expected revenue of the optimal policy obtained by the EM-C algorithm is slightly better than that of the plug-in policy for the discrete problem. To demonstrate convergence of the EM-C algorithm,
Figure \ref{fig:single_airline_combined}
 shows the objective function values of the EM-C algorithm over 5 iterations for the discrete problem \eqref{equ:prob_dyna_pricing_single_prod} when $n^c = 20, 10, \text{and}\ 5$ respectively.

\begin{figure}[htbp]
\centering
\subfloat[][$n_c=20$]{
\label{fig:single_airline_combined_n_20}
\includegraphics[keepaspectratio=true, width=0.7\textwidth, clip=true]{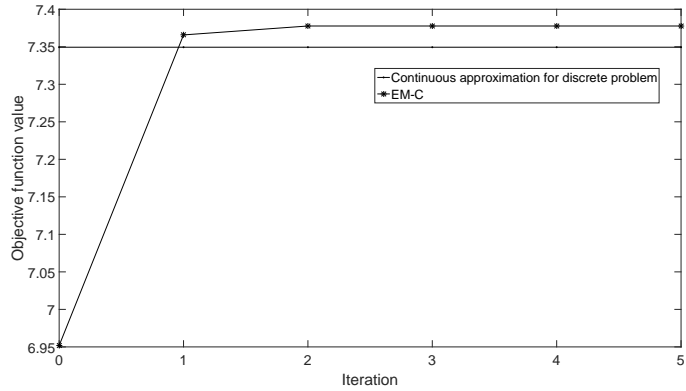}
}
\\
\subfloat[][$n_c=10$]{
\label{fig:single_airline_combined_n_10}
\includegraphics[keepaspectratio=true, width=0.7\textwidth,clip=true]{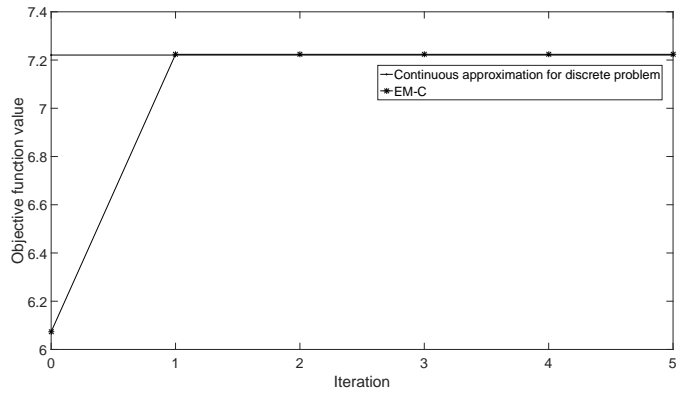}
}
\\
\subfloat[][$n_c=5$]{
\label{fig:single_airline_combined_n_5}
\includegraphics[keepaspectratio=true, width=0.7\textwidth,clip=true]{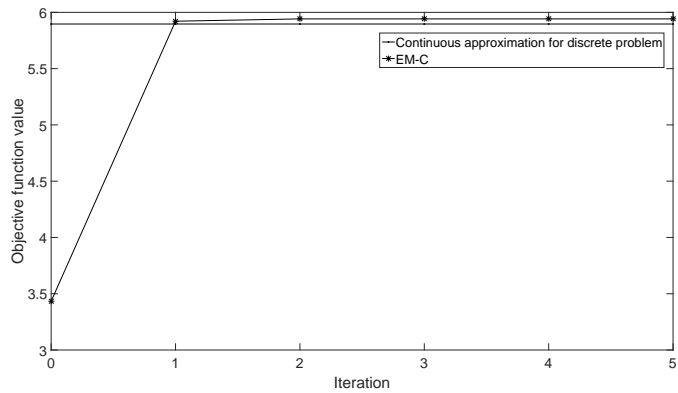}
}
\\
\caption[]{The objective function values (expected revenues) obtained by the EM-C algorithm over 5 iterations for the single product monopoly pricing problem \eqref{equ:prob_dyna_pricing_single_prod}. The EM-C algorithm converges after 2 iterations. Each iteration takes around 3 minutes.
%(a) When $n^c = 20$, the optimal objective function values obtained by the EM-C algorithm is $7.3777(0.0270)$. The objective function values obtained by the continuous approximation policy for the discrete problem is $7.3494(0.0271)$. (b) When $n^c = 10$, the optimal objective function values obtained by the EM-C algorithm is $7.2237(0.0260)$. The objective function values obtained by the continuous approximation policy is $7.2207(0.0257)$. (c) When $n^c = 5$, the optimal objective function values obtained by the EM-C algorithm is $5.9419(0.0204)$. The objective function values obtained by the continuous approximation policy for the discrete problem is $5.8964(0.0205)$.
}
\label{fig:single_airline_combined}
\end{figure}

\subsection{Multi-Product Case}

\subsubsection{The Multi-product Monopoly Pricing Model}
We extend the single product monopoly pricing model into a multi-product
model as first studied in \cite*{gallego1997multiproduct}. With higher dimension, this
problem cannot be solved analytically.
%Various numerical methods have proposed to tackle this problem.
%
%Our algorithm is not necessarily
%superior in terms of numerical efficiency, rather, we use this example
%to illustrate the real life application of our algorithm, and to examine
%the convergence results.
More precisely, suppose the airline flight network has $n_{l}$ legs (direct flights), based on which
there are $n_{i}$ itineraries. Define a matrix $A:=[a_{kj}]\in\mathbb{R}^{n_{l}\times n_{i}}$, where $a_{kj}\in\{0, 1\}$ and $a_{kj}=1$ if and only if the direct flight $k$ is a part of the itinerary $j$. For example, consider a simple network with 3 nodes, $\{1, 2, 3\}$, two direct flights
$\{1\to2, 2\to3\}$, and three itineraries $\{1\to2, 2\to3, 1\to2\to3\}$. Then for this flight network,
\begin{equation}\label{equ:A}
A= \begin{pmatrix}
1&0&1\\
0&1&1
\end{pmatrix}.
\end{equation}
As one can see that the dimension of this problem increases very quickly, Monte Carlo methods might offer a realistic hope for solving such a problem.

Let $p\in\mathbb{R}^{n_{i}}$ be the vector of prices for the $n_i$ itineraries. The
customers who need the $n_i$ itineraries come to buy tickets according to the process $N^{\lambda}\in\mathbb{N}^{n_{i}}$ with arrival rates $\lambda\in\mathbb{R}^{n_{i}}$. $p$ is assumed to be a function of the customer arrival rates $\lambda$.
Let the initial capacities of the direct flights be $n^{c}\in\mathbb{N}^{n_{l}}$. The objective is to optimize the expected revenue by choosing
the prices $p$, or equivalently, the customer arrival rates $\lambda$. More precisely, the multi-product monopoly pricing problem is formulated as
{\allowdisplaybreaks
\begin{align}
V(n^c, T) = \sup_{\lambda_s}\ \ & E_0\left[\int_{0}^{T}p(\lambda_{s})'dN_{s}^{\lambda}\right]\label{equ:multi_product_problem}\\
\text{s.t.}\ \ & V(n, 0)=V(0, t)=0,\;\forall n\in\mathbb{N}^{n_l}, \forall t>0,\notag\\
& \int_{0}^{T}AdN_{s}^{\lambda}\leq n^{c},\notag\\
& p(\lambda_{s})_j=(\epsilon_{0, j}^{-1}\log\frac{\lambda_{0, j}}{\lambda_{s, j}}+1)p_{0,j},\mbox{ for }s\leq T, j=1,\ldots, n_i.\notag
\end{align}}%
%The corresponding HJB equation is
%{\allowdisplaybreaks
%\begin{align*}
%\frac{\partial V^{*}(n^{c}, s)}{\partial s} & = \sup_{\lambda\in\Lambda(x,s)}\left[p_{s}'\lambda_{s}-\sum_{j=1}^{n^{c}}\lambda^{j}(V^{*}(n^{c}, s)-V^{*}(n^{c}-A^{j}, s))\right]\\
%V^{*}(n^{c}, t) & = 0,\;\forall t,n^{c}:n^{c,k}<a_{kj}\mbox{ for some }k\mbox{ and all }j=1,..,n_i\\
%V^{*}(n^{c}, 0) & = 0,\;\forall n^{c}\\
%\Lambda(s,n^{c}) & = \Lambda(s)\cap\left\{ \lambda:\lambda^{j}=0\; if\; a_{kj}<n^{c,k}\mbox{ for some }k\right\},
%\end{align*}}%
%where $\Lambda(s)$ is the feasible policy set and $A^{j}$ denotes the $j$th column of $A$.

As the high dimensional HJB equation corresponding to the problem \eqref{equ:multi_product_problem} is difficult to solve, \cite*{gallego1997multiproduct}
provide two heuristic policies called MTS and MTO
that are asymptotically optimal as the size of the problem
goes to infinity. Both heuristic policies
use the optimal control from a deterministic version of this problem,
which assumes that the control $\lambda_t$ is time invariant and deterministic.
The deterministic case is solved as a constrained non-linear optimization
problem. Denote the corresponding control and price as $\hat\lambda^{*}$ and $\hat p^{*}$ respectively.
More precisely,  the MTS and MTO policies are given below:
\begin{enumerate}
\item[(i)] MTS policy: set the prices equal to the deterministic optimal price
$\hat p^{*}$ and pre-allocate seats for each itinerary accordingly. Stop
selling the ticket of itinerary $j$ if the pre-allocated seats for itinerary $j$ are exhausted;
\item[(ii)] MTO policy: set the prices equal to the deterministic optimal price
$\hat p^{*}$ and sell tickets in the order of customer
arrival. Stop selling the ticket of itinerary $j$ when the inventory of
at least one direct flight $k$ drops strictly below $a_{kj}$.
\end{enumerate}
%The two heuristic policies differ as follows: for MTS policy, it may
%occur that at some point a certain itinerary has been sold out. The
%sale of the itinerary stops even if seats are still available from
%the residual capacity of other itineraries. This is allowed by MTO
%policy.

We focus on a discrete-time setting of the problem. The time horizon $[0, T]$ is divided
into $n_{T}$ equal periods, denoted as $t_{0}=0 < t_1 < \cdots < t_{N_{T} - 1} < t_{N_{T}}=T$. The discrete time problem is formulated as
{\allowdisplaybreaks
\begin{align}
\max_{c_{t_k,j},k=0,\ldots,n_T-1,j=1,\ldots,n_i}\ \ & E_0\left[\sum_{k=0}^{n_{T}-1}p(\lambda_{t_{k}})'(N_{t_{k+1}}^{c}-N_{t_{k}}^{c})\right]\label{equ:dyna_pric_multi_prod_formulation}\\
\text{s.t.}\ \quad\quad\quad\ \ \  & N_{t_{k+1},j}^{\lambda}-N_{t_{k},j}^{\lambda}\sim Poisson(\lambda_{t_{k},j}T/n_T), j=1, \ldots, n_i, \forall k\nonumber\\
  & N_{t_{k+1}}^{c}=G(n^{c},N_{t_{k}}^{c}, N_{t_{k+1}}^{\lambda}-N_{t_{k}}^{\lambda}), \forall k,\label{equ:G_func}\\
  & p(\lambda_{t_{k}})_j=(\epsilon_{0,j}^{-1}\log\frac{\lambda_{0, j}}{\lambda_{t_{k},j}}+1)p_{0, j}, j=1, \ldots, n_i, \forall k,\nonumber\\
  % & \lambda^j_{t_{i}}=\frac{\lambda_{0,j}e^{\epsilon_{0,j}}}{1+\exp(c^j_{t_i})}, j=1, \ldots, n_i,\\
  & \lambda_{t_{k},j}=\min(\lambda_{0,j}e^{\epsilon_{0,j}}, \max(c_{t_k,j}, 0)), j=1, \ldots, n_i, \forall k, \label{equ:c_constraint}\\
  & c_{t_k, j} \in \mathbb{R}, j=1, \ldots, n_i.\nonumber
\end{align}}%
%where
%\be\label{equ:lambda_pricing_multi}
%\lambda^j_{t_{i}}=\frac{\lambda_{0,j}e^{\epsilon_{0,j}}}{1+\exp(c^j_{t_i})}.
%\ee
In the formulation, $\lambda_{t_{k},j}$ should satisfy the constraint that $0<\lambda_{t_{k},j}<\lambda_{0,j}e^{\epsilon_{0,j}}$.
The constraint is imposed by \eqref{equ:c_constraint}, which means that $\lambda_{t_{k},j}=c_{t_{k},j}$ if $0<c_{t_{k},j}<\lambda_{0,j}e^{\epsilon_{0,j}}$, and $\lambda_{t_{k},j}=0$ if $c_{t_{k},j}\leq 0$, and $\lambda_{t_{k},j}=\lambda_{0,j}e^{\epsilon_{0,j}}$ if $c_{t_{k},j}\geq \lambda_{0,j}e^{\epsilon_{0,j}}$.\footnote{In the implementation, we actually uses $\lambda_{t_{k},j}=\min((1-\delta)\lambda_{0,j}e^{\epsilon_{0,j}}, \max(c_{t_k,j}, \delta\lambda_{0,j}e^{\epsilon_{0,j}}))$ in order to ensure that $0<\lambda_{t_{k},j}<\lambda_{0,j}e^{\epsilon_{0,j}}$, where $\delta=10^{-5}$.}
%+\lambda_{0,j}e^{\epsilon_{0,j}}1_{c^j_{t_{k}}>\lambda_{0,j}e^{\epsilon_{0,j}}}$.
%is for imposing the constraint that .
The control of the problem is $c_{t_k}=(c_{t_k, 1}, \ldots, c_{t_k, n_i})'$. The state variables of the problem are
the residual capacities $R_{t_{k}}=n^{c}-AN_{t_{k}}^{c}$.

Similar to the single product case, we cap the customer arrival process at $n^c$ to impose the capacity constraint. The capping becomes more complicated
in the multi-product case, since there can be more than one way to
allocate the remaining capacity of a direct flight to the itineraries. As a result, the function $G$ in \eqref{equ:G_func} is defined as
\begin{equation}\label{equ:def_G}
G(n^{c}, N_{t_{k}}^{c}, N_{t_{k+1}}^{\lambda}-N_{t_{k}}^{\lambda}) := \begin{cases}
N_{t_{k}}^{c} + N_{t_{k+1}}^{\lambda}-N_{t_{k}}^{\lambda}, & \mbox{ if }A(N_{t_{k}}^{c} + N_{t_{k+1}}^{\lambda}-N_{t_{k}}^{\lambda})\leq n^{c},\\
N^{c}_{t_{k}}+\Delta_{t_{k+1}}^{c}, & \mbox{ otherwise,}
\end{cases}
\end{equation}
where
\begin{align*}
\Delta_{t_{k+1}}^{c} =\arg  \max_{N}\ \ & p(\lambda_{t_{k}})'N\\
 \text{s.t.}\ \  &  AN\leq R_{t_{k}}, N\geq 0, N \leq N_{t_{k+1}}^{\lambda}-N_{t_{k}}^{\lambda}, N\in\mathbb{N}^{n_i}.
\end{align*}
The condition $A(N_{t_{k}}^{c} + N_{t_{k+1}}^{\lambda}-N_{t_{k}}^{\lambda})\leq n^{c}$ in \eqref{equ:def_G} is the case when capacity is not
exceeded, under which no capping is performed. If the capacity is
exceeded for some direct flights, then the residual capacities are allocated optimally to maximize
the revenue in the period $[t_k, t_{k+1}]$. This suggests that when tickets are
about to be sold out, the remaining seats will be allocated to those itineraries
that generate more revenue.

%Verification of Assumption \ref{assump:1} is similar for Example
%\ref{sub:inv_single}.
%
%Similar to the single-product case, we have the following proposition
%similar to Proposition \ref{prop:cond_single_inv}.
%\begin{prop}
%If we extend the state space $S$ to be $\mathbb{R}_{+}^{n_{l}}$,
%then Assumption \ref{assump:1} is satisfied with (5) and (6) replaced
%with (5') and (6').
%\end{prop}

\subsubsection{Numerical Results}

We consider a particular case of problem \eqref{equ:dyna_pric_multi_prod_formulation} in which the flight network has 3 nodes, $\{1, 2, 3\}$, two direct flights
$\{1\to2, 2\to3\}$, and three itineraries $\{1\to2, 2\to3, 1\to2\to3\}$. Suppose the capacities
of the direct flights are $n^c = (n_{1\to2}^{c}, n_{2\to3}^{c})'=(300, 200)'$. Suppose $p_{0}=(p_{0,j})=(220, 250, 400)'$, $\epsilon_{0}=(\epsilon_{0,j})=(1.0, 1.2, 1.1)'$ and $\lambda_{0}=(\lambda_{0,j})=(300,300,300)'$. Let $T=1$ and $n_{T}=6$.
The state variables are the residual capacity $R=(R_{1\to2},R_{2\to3})'=n^{c}-AN^c$, where the matrix $A$ is given in \eqref{equ:A}.
We use linear functions of the state variables as the basis functions for the controls
$c=(c_{1\to2},c_{2\to3},c_{1\to2\to3})'$, i.e., the basis functions are
{\allowdisplaybreaks
\begin{align*}%\label{eq:basis_mult_inv}
&\phi_{1,1\to2}(R)=(1,0,0)',&&\phi_{1,2\to3}(R) =(0,1,0)',&&\phi_{1,1\to2\to3}(R)=(0,0,1)',\\
&\phi_{2,1\to2}(R)=(R_{1\to2},0,0)',&&\phi_{2,2\to3}(R)=(0,R_{1\to2},0)',&&\phi_{2,1\to2\to3}(R)=(0,0,R_{1\to2})',\\
&\phi_{3,1\to2}(R)=(R_{2\to3},0,0)',&&\phi_{3,2\to3}(R)=(0,R_{2\to3},0)',&&\phi_{3,1\to2\to3}(R)=(0,0,R_{2\to3})'.
\end{align*}}%
We denote the control parameter at period $t$ as $\theta_t=(\theta_{t,k,l})_{k=1,2,3, l\in\{1\to2, 2\to3, 1\to2\to3\}}$. Then, the control $c_t$ is
\be
c_t=\sum_{k=1}^3\sum_{l\in \{1\to2, 2\to3, 1\to2\to3\}}\theta_{t, k, l}\phi_{k,l}(R_t).\notag
\ee
We then apply the EM-C algorithm to the problem. We use $N=10,000$ sample paths in the simulation and use $m=2,000$
iterations in the SA algorithm. The initial control parameters $c^0_0$ and $\theta_t^0$
are set to be %equal to the deterministic MTO policy, i.e.,
$c^0_0=(100, 100, 100)'$ and $\theta_{t,1,l}^{0}=100, \theta_{t,2,l}^{0}=\theta_{t,3,l}^{0}=0, \forall l, \forall t$.

%Each iteration takes about 1.3 hours under a Matlab program. The bottleneck of the program is the iteration of the SA algorithm that need to be implemented by ``for loops'' in Matlab, which is known to be slow. The computational time can be greatly reduced if the algorithm is implemented by a compiled language such as C/C++. The optimal revenue  obtained by the EM-C algorithm is $1.8729\times 10^5$ (with standard error 54.7). The standard error is equal to the sample standard deviation of the $N$ samples on the right-hand side of \eqref{equ:appro_E_0} divided by $\sqrt{N}$.
%We can see the performance of our algorithm (MC Optimal) reaches some
%stable point roughly in 2 rounds. Coefficients $\theta_{2,12}$ and
%$\theta_{3,23}$ at time 1/2 and 1 show the same pattern.

%\begin{figure}[h]
%\begin{centering}
%\includegraphics[width=\textwidth]{fig_dyna_pricing_multi_prod_parallel_N_100000_n_1_m_2000_n_rep_1_n_for_obj_eval_100000_with_MTS_MTO.eps}
%\end{centering}
%\caption{Objective function of dynamic multi-product pricing inventories over 6 iterations. The algorithm converged after 5 iterations. It uses 100000 sample paths in the master simulation. It takes 10 hours for each iteration. The optimal revenue is $1.876\times 10^5$ (with standard error 19). \label{fig:stoch_inv_mul_opt_rev_iter}}
%\end{figure}

\begin{figure}[htbp]
\begin{centering}
\includegraphics[width=\textwidth]{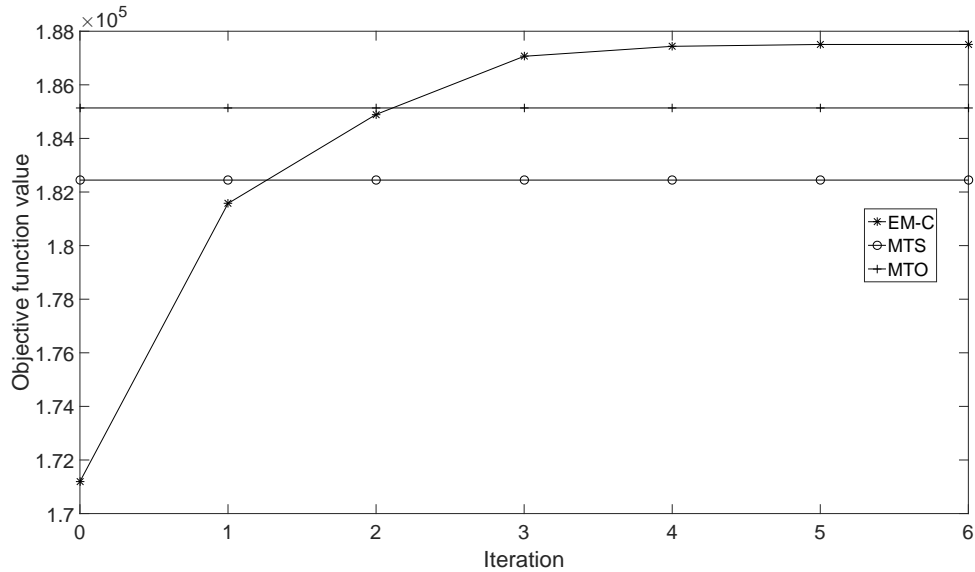}
\end{centering}
\caption{Objective function values of two heuristic methods, MTO and MTS, and the (rigorous) EM-C algorithm. The EM-C algorithm converged after 5 iterations. It uses $N=10,000$ sample paths in the simulation and $m=2,000$ iterations in the SA algorithm. It takes 1.3 hours for each iteration under a Matlab program. The bottleneck of the program is the iteration of the SA algorithm that need to be implemented by ``for loops'' in Matlab, which is known to be slow. The computation time can be greatly reduced if the algorithm is implemented by a compiled language such as C/C++. The optimal revenue obtained by the EM-C algorithm is $187292.9$ (with standard error 54.7). The standard error is equal to the sample standard deviation of the $N$ samples on the right-hand side of \eqref{equ:appro_E_0} divided by $\sqrt{N}$. \label{fig:stoch_inv_mul_opt_rev_iter}}
\end{figure}

Figure \ref{fig:stoch_inv_mul_opt_rev_iter} shows the objective function values of the EM-C algorithm over 6 iterations.
The EM-C algorithm converged after 5 iterations. It appears that the (rigorous) EM-C algorithm yields a much higher revenue than that from the two heuristic algorithms MTO and MTS. Table \ref{tab:inv_mult_rev} compares the distributions of revenues obtained by the EM-C algorithm, MTO, and MTS, respectively, using $N=10,000$ sample paths in the simulation.
The distribution of the total revenue under the EM-C algorithm has higher mean, higher skewness, smaller kurtosis, and higher quantile (at 1\%, 5\%, 95\%, 99\% level) than that under the MTO and MTS. Table \ref{tab:inv_mult_rev} also compares the revenues at the 3rd period and the 6th period obtained by the EM-C algorithm, MTO, and MTS, respectively. At the 3rd period, the EM-C algorithm performs similarly to MTO and MTS; however, at the 6th period, the EM-C algorithm performs better than the other two in terms of mean and standard error.

The total revenue generated by the EM-C method is 187,292.9 with standard error 54.7; while the two standard heuristic methods (MTO and MTS) give 185,090.2 and 182,433.5 with standard error 58.2 and 59.0 respectively. Thus, the EM-C method leads to an expected revenue increase of 1.2\% and 2.7\%, respectively. This is a very significant improvement, in view of the tight margin of airlines with large revenue and small profits.\footnote{For example, in 2015 Singapore Airlines had the revenue of
\$15,228 million, but the profit was \$801 million, which was only 5.26\% of the revenue. If the dynamic pricing of tickets can increase the revenue by 1\% without incurring additional cost, then it would lead to a significant increase in profit.}

Figure \ref{fig:Scatter-opt_rev} compares the histogram of the total revenue obtained under the EM-C algorithm, MTO, and MTS; the EM-C algorithm achieves a better right tail distribution than the other two policies.

\begin{table}[htbp]
\centering
\resizebox{\linewidth}{!}{
\begin{tabular}{rrrrrrrrrr}
\hline
\hline
\addlinespace
 & \multicolumn{3}{c}{total revenue} & \multicolumn{3}{c}{revenue at 3rd period} & \multicolumn{3}{c}{revenue at 6th period}\\
\addlinespace
\cline{2-4}\cline{5-7}\cline{8-10}
\addlinespace
 & EM-C & MTO & MTS & EM-C & MTO & MTS & EM-C & MTO & MTS\\
%\hline
    mean  & 187292.9 & 185090.2 & 182433.5 & 31528.0 & 31669.5 & 30815.0 & 30199.0 & 26641.0 & 24655.8 \\
    stderr & 54.7  & 58.2  & 59.0  & 41.7  & 42.4  & 41.7  & 37.6  & 61.3  & 56.0 \\
    skewness & -0.31 & -1.42 & -0.99 & 0.16  & 0.15  & 0.18  & -0.30 & -0.74 & -0.31 \\
    kurtosis & 3.12  & 5.05  & 3.75  & 2.96  & 3.06  & 2.99  & 3.02  & 3.89  & 2.97 \\
    1\% quantile & 173321.7 & 166656.9 & 165253.7 & 22433.5 & 22389.7 & 21956.9 & 20804.8 & 8437.2 & 10390.9 \\
    5\% quantile & 177686.2 & 173154.8 & 170998.6 & 24893.0 & 24884.1 & 24174.8 & 23695.9 & 15162.5 & 14674.2 \\
    95\% quantile & 195699.0 & 190570.7 & 189292.2 & 38556.0 & 38833.6 & 37934.1 & 36028.7 & 35306.6 & 33237.3 \\
    99\% quantile & 198886.7 & 190958.9 & 189292.2 & 41559.2 & 41917.8 & 41113.9 & 37924.7 & 38364.4 & 36544.1 \\
\addlinespace
\hline
\hline
\end{tabular}}
\caption{Multi-product monopoly pricing: comparing the distributions of revenues obtained by the EM-C algorithm, MTO, and MTS respectively, using $N=10,000$ sample paths in the simulation. ``Std. error'' indicates the standard error of the mean estimate.
The distribution of the total revenue under the EM-C algorithm has higher mean, higher skewness, smaller kurtosis, and higher quantile than that under the MTO and MTS. At the 3rd period, the EM-C algorithm performs similarly to MTO and MTS; however, at the 6th period, the EM-C algorithm performs better than the other two in terms of mean and standard error.\label{tab:inv_mult_rev}}
\end{table}

\begin{figure}[htbp]
\centering
\includegraphics[width=0.33\textwidth]{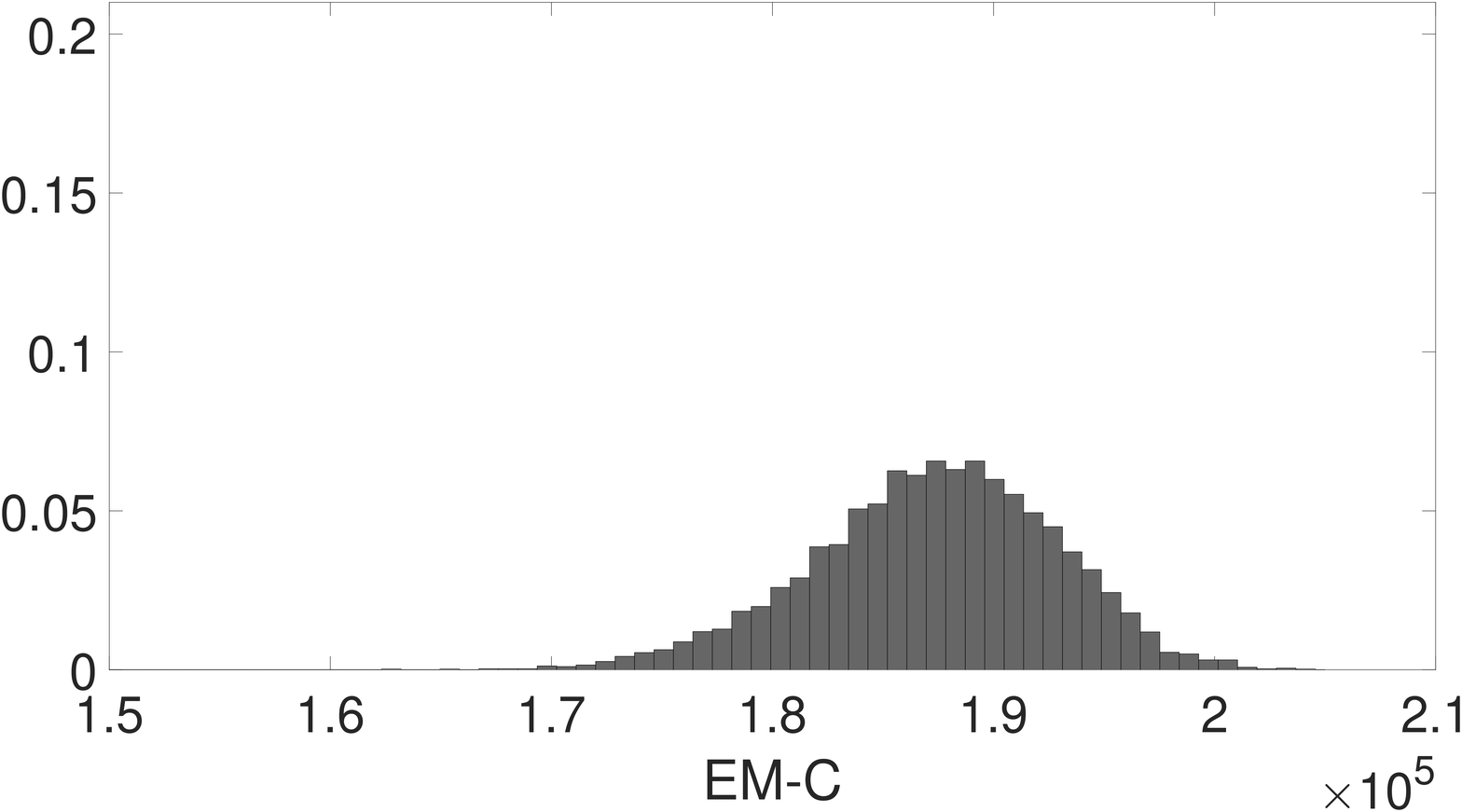}\includegraphics[width=0.33\textwidth]{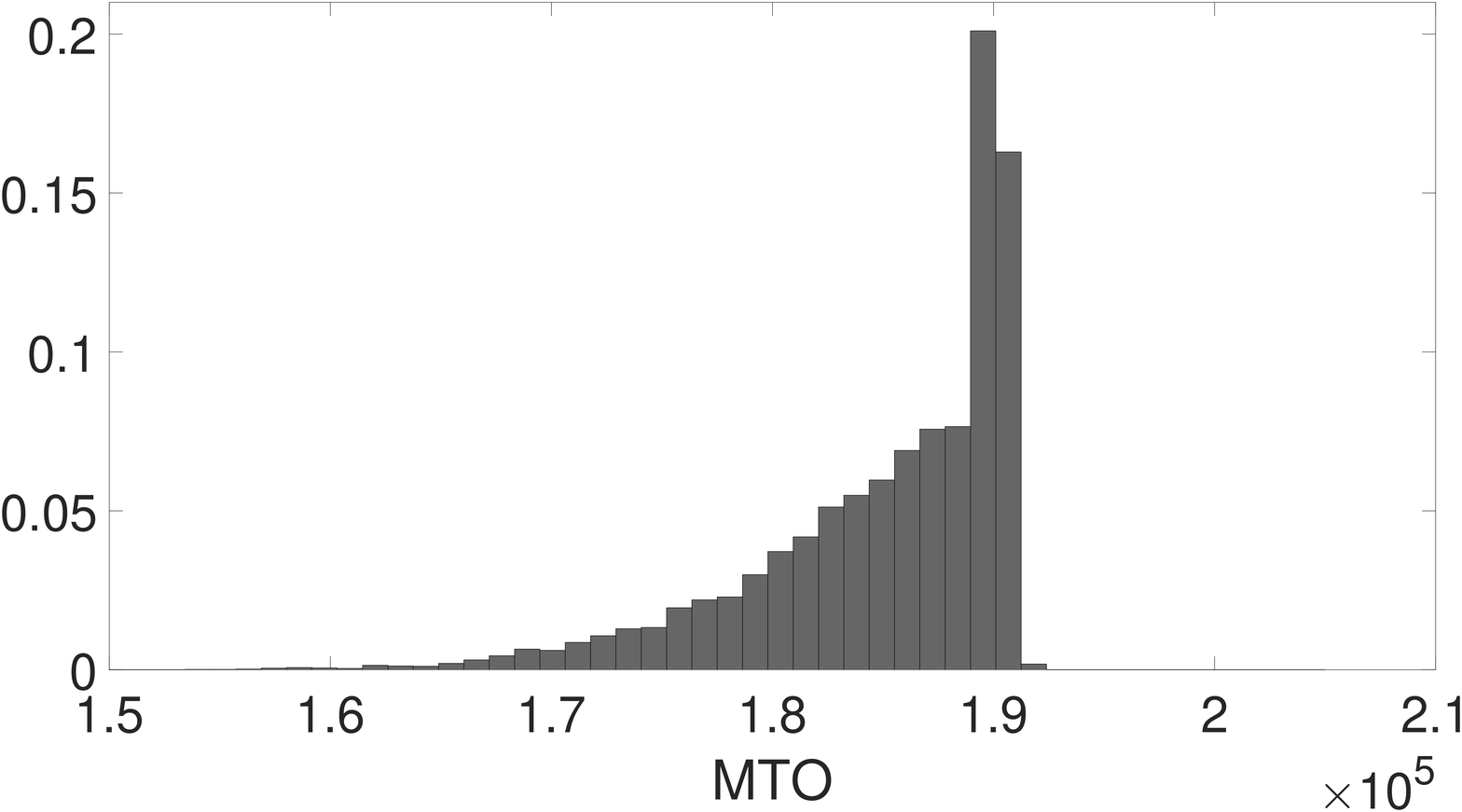}\includegraphics[width=0.33\textwidth]{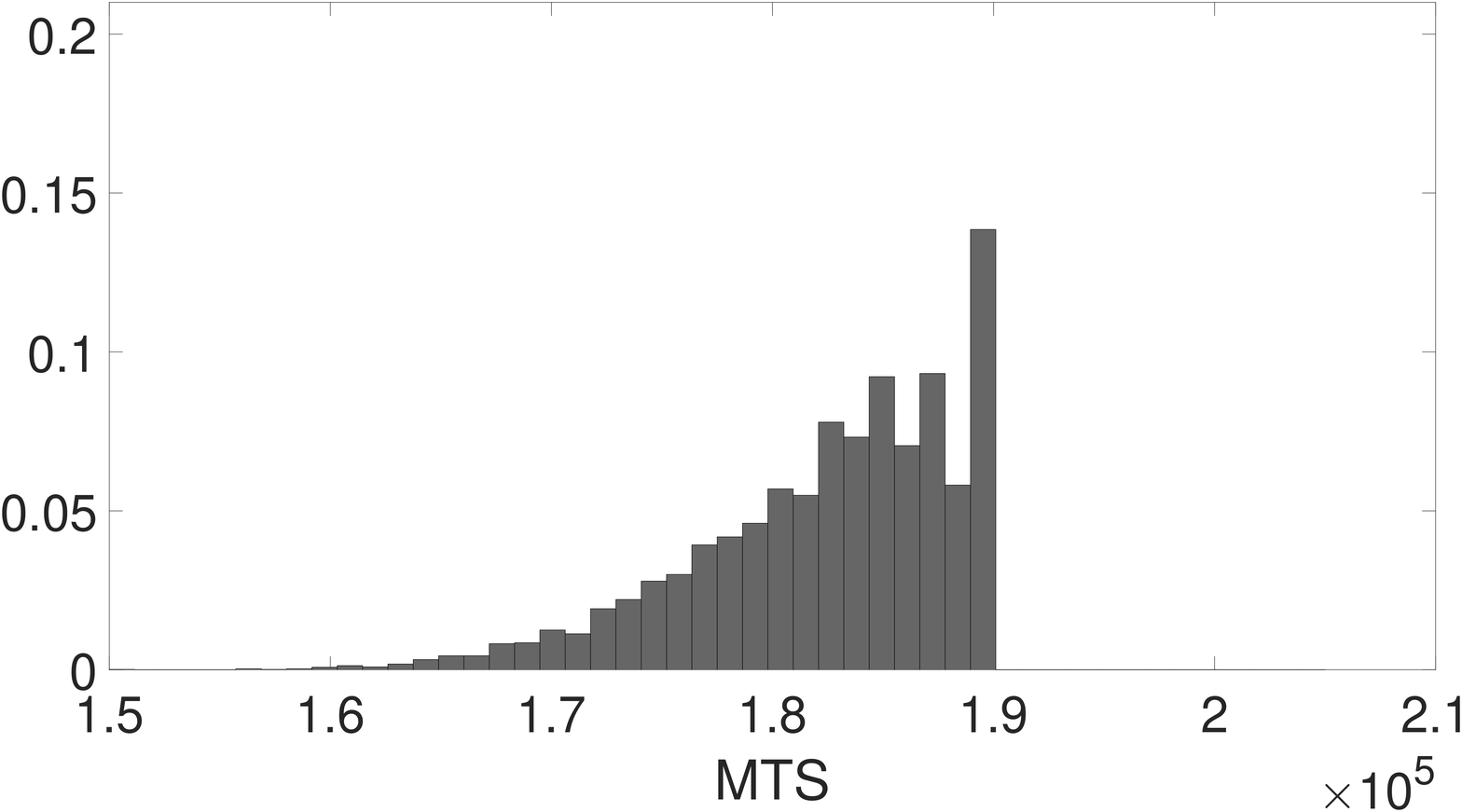}
\caption{Multi-product monopoly pricing: comparing the histograms of total revenue
obtained by EM-C algorithm (left), MTO (middle), and MTS (right). The histograms are based on 10,000 sample path of simulations.}
 \label{fig:Scatter-opt_rev}
\end{figure}

Figure \ref{fig:opt_dyn_c} compares the ticket pricing functions at the beginning of the 3rd period (i.e., at time $t_{2}=1/3$) and the beginning of the 6th period (i.e., at time $t_5=5/6$) under the EM-C algorithm with those under MTO/MTS, which are constant prices that do not change with the residual flight capacities. The prices under the MTO and those under the MTS are the same, although the two algorithms adopt different policies to allocate residual capacities of direct flights to itineraries.
% optimal controls $p$ and deterministic
%policy $p^{*,d}$ are plotted with respect to the residual capacity
%for time $t_{i}=1/2$ and 1.
Comparing the ticket pricing functions at the beginning of the 3rd period with those at the beginning of the 6th period, we
can see that the prices at the beginning of the 6th period under the EM-C algorithm are more sensitive with respect
to the residual capacities than those at the beginning of the 3rd period; this is reasonable as the optimal ticket prices should be more dependent on the residual capacity, to maximize revenue when the time left for the sale of the tickets is only one period.
%Also, the dependence of MC optimal $p$ on residual capacity
%is especially significant for $\theta_{12}^{2}$ and $\theta_{23}^{3}$,
%i.e., $p_{12}$ on $R_{12}$ and $p_{23}$ on $R_{23}$.
%This is reasonable
%since when closer to the end of the horizon and residual capacity
%is small, the randomness in the arrival process makes the ratios among
%residual capacity deviate farther away from the original ones, hence
%the deterministic optimal control becomes less favorable. As a result,
%a re-optimization becomes necessary. While early in time and residual
%capacity is large, the ratios among residual capacity stays relatively
%stable, hence the deterministic optimal control is a reasonably good
%policy.

\begin{figure}[htbp]
\centering
\includegraphics[width=\textwidth]{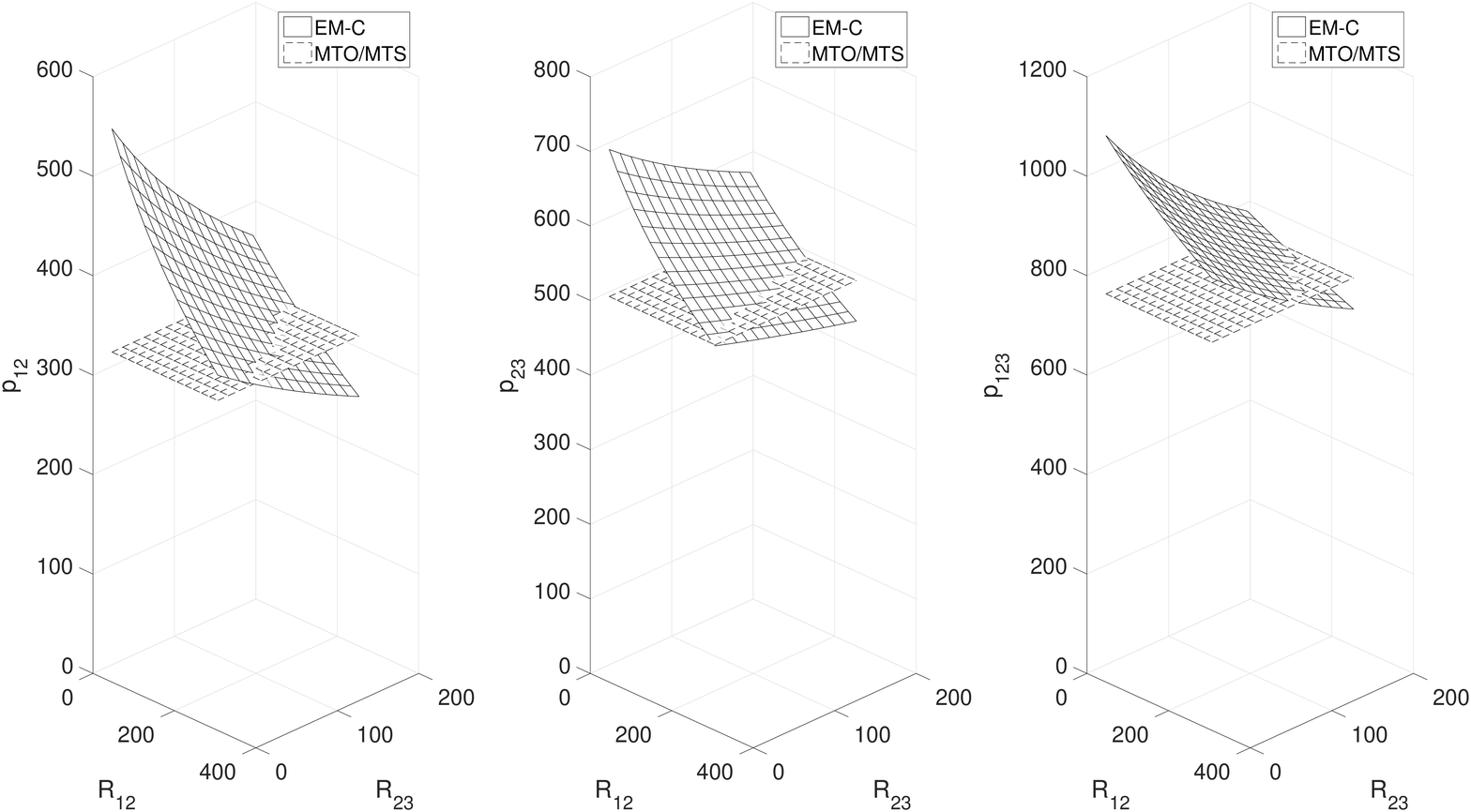}
\centering
\includegraphics[width=\textwidth]{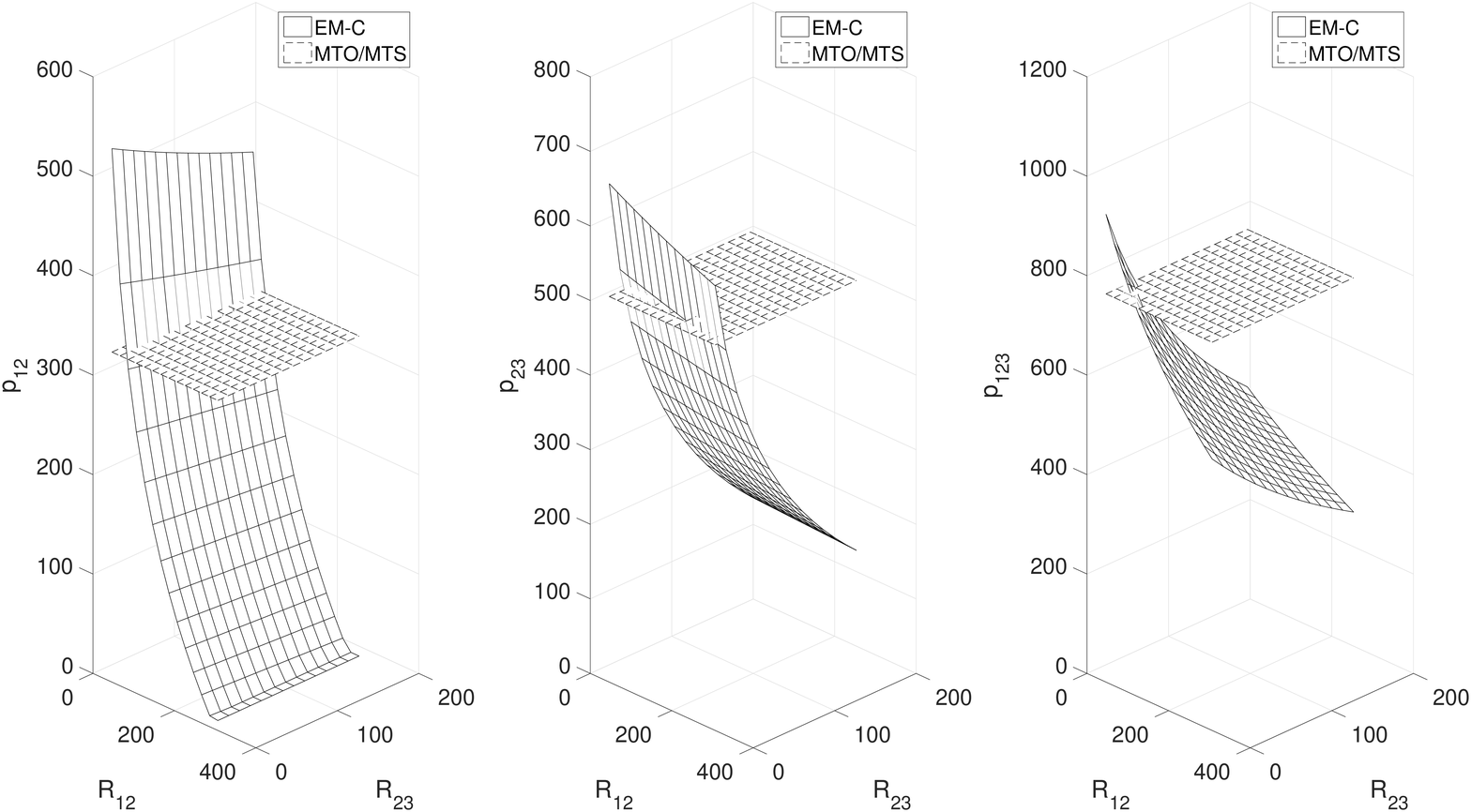}
\caption{Multi-product monopoly pricing: the airline ticket prices as functions of the residual capacities obtained under the EM-C algorithm and those under MTO/MTO. The first row plots the prices at the beginning of the 3rd period (i.e., at time $t_2=1/3$), and the second row plots the prices at the beginning of the 6th period (i.e., at time $t_{5}=5/6$). The prices under the MTO and those under the MTS are the same, although the two algorithms adopt different policies to allocate residual capacities of direct flights to itineraries.
%Each plot shows the price for certain itinerary with respect
%to residual capacities.
$R_{12}$ and $R_{23}$ denote the residual capacities of the direct flight $1\to 2$ and $2\to 3$ respectively.
$p_{12}$, $p_{23}$, and $p_{123}$ denote the airline ticket prices for the itinerary $1\to2$, $2\to3$, and $1\to2\to 3$ respectively.}
\label{fig:opt_dyn_c}
\end{figure}

%\subsubsection{Incorporating heuristics}
%
%We can also incorporate heuristics into the forms of policies. For
%example, if we have the intuition that certain aspects of the policy
%only depend on a subset of states, we further restrict our policy
%space to achieve that.
%
%In our simple example, based on our approximated optimal control,
%we simplify our policies for itineraries $1-2$ and $2-3$ to only
%depend on residual capacity of its corresponding airline, respectively,
%while the policy for $1-2-3$ depends on the minimum of the residual
%capacities of the two airlines. That is, we restrict our basis to
%the following
%\[
%\begin{cases}
%\phi_{12}^{1}(R) & =[1,0,0]'\\
%\phi_{12}^{2}(R) & =[R_{12},0,0]'
%\end{cases}\quad\begin{cases}
%\phi_{23}^{1}(R) & =[0,1,0]'\\
%\phi_{23}^{3}(R) & =[0,R_{23},0]'
%\end{cases}\quad\begin{cases}
%\phi_{123}^{1}(R) & =[0,0,1]'\\
%\phi_{123}^{4}(R) & =[0,0,\min(R_{12},R_{23})]'
%\end{cases}
%\]
%
%
%In the first plot of Figure \ref{fig:stoch_inv_mul_opt_rev_iter},
%mean revenue of the optimal policies with this basis are illustrated.
%With the new basis, the corresponding optimal policy still achieves
%comparable good performance improvement. In the first round, it brings
%similar improvement, while it converges to a stable policy relatively
%slower than the one with the original basis in (\ref{eq:basis_mult_inv}).
%After the fourth round, it reaches a similar level as the one using
%(\ref{eq:basis_mult_inv}).
%
%As a result, when heuristics are used properly, we can reach comparable
%performance improvement, while significantly reduce the dimension
%of the policy.

\section{Application 2: Real Business Cycle}\label{sec:business_cycle}

In this section, we apply the EM-C algorithm to study the problem of real business cycle
(see e.g., \cite*{kydland1982time}, \cite*{long1983real}, \cite*{hansen1985indivisible}, and \cite*{christiano1990linear}).
In the literature this is typically studied assuming infinite time horizon,
under which a stationary solution can be computed. In particular,
a log-linear linear-quadratic (LQ) approximation is used to approximate
the objective function, which transforms the problem to a well-studied
linear-quadratic programming problem. However, by using the EM-C algorithm, we show that there are very significant differences between
the finite time horizon and infinite time horizon problem.
Indeed, the policies used for the infinite horizon
problem can be very different from those for the finite time horizon problem, even if we take a 10-year time horizon; and our
algorithm yields much higher expected utility and more sensible control policy than the log-linear LQ method in the finite time horizon problem.

\subsection{The Model}

The standard infinite horizon problem in the literature is as follows
{\allowdisplaybreaks
\begin{align}\label{eq:real_bus_cyc}
  \max_{g_t}\ \   & E_0\left[\sum_{t=0}^{\infty}\beta^{t}u(k_{t},k_{t-1},x_{t})\right]=E_0\left[\sum_{t=0}^{\infty}\beta^{t}\frac{g_{t}^{1-\tau}}{1-\tau}\right]\\
 \text{s.t.}\ \  & x_{t+1} =\rho x_{t}+\epsilon_{t+1}, t \geq 0,\notag\\
 & k_{t}=\exp(x_{t})k_{t-1}^{\gamma}-g_{t}+(1-\delta)k_{t-1}, t \geq 0,\notag\\
 & g_{t}\in[0,\exp(x_{t})k_{t-1}^{\gamma}+(1-\delta)k_{t-1}], t \geq 0,\notag
\end{align}}%
where $(k_{-1}, x_0)$ is given as the initial state at period $t=0$; $x_{t}$ is the technology innovation level at period $t$, which evolves following
a time-series AR(1) model; $\exp(x_{t})k_{t-1}^{\gamma}$ is the total
production at period $t$; $g_{t}$ is the consumption at period $t$;
$k_{t}$ is the end-of-period-$t$ capital, which depends on the depreciation rate of capital $\delta$; $\tau\in(0,1)$ is the risk preference parameter. The logarithmic preference can be considered
as the limiting case when $\tau\rightarrow1-$. The state of the model at period $t$ is
$s_t = (k_{t-1},x_{t}).$

The main idea of log-linear LQ approximation is to approximate the
objective function with linear or quadratic functions, so that the
approximated problem fits into the linear quadratic programming framework,
which is analytically tractable. Let $\tilde k_{t}=\log(k_{t})$. The log-linear LQ approximation approach applies
a second-order Taylor series expansion to $\tilde u(\tilde k_{t}, \tilde k_{t-1},x_{t}):=u(\exp(\tilde k_{t}),\exp(\tilde k_{t-1}),x_{t})$
with respect to $(\tilde k_{t}, \tilde k_{t-1}, x_t)$ about $(\log k^*, \log k^*, x^*)$,
where $k^{*}$ and $x^{*}$ are the steady-state values of $k_t$ and $x_t$ of the non-stochastic version of \eqref{eq:real_bus_cyc} obtained by setting $\epsilon_t = 0$ for all $t$.
%Then, the first-order necessary condition for $k_t$ to solve the log-linear LQ problem is
%\begin{equation}\label{eq:first_order_log_linear}
%   E_t[\log k_{t+1}-\log k^{*}]-\phi(\log k_{t}-\log k^{*})+\frac{1}{\beta}(\log k_{t-1}-\log k^{*})+\frac{q}{\beta k^{*}}x_{t}=0,
%\end{equation}
More precisely, the log-linear LQ approximation policy to the infinite horizon problem \eqref{eq:real_bus_cyc} is given by \citet[][Eq. (2.19)]{christiano1990linear}
\begin{equation*}
k_t = (k^*)^{(1-\lambda)}\exp\left[\frac{q}{k^*}\frac{\lambda}{1-\beta \rho\lambda}x_t\right]k_{t-1}^{\lambda}, t\geq 0,\ \text{where}
\end{equation*}
{\allowdisplaybreaks
\begin{align}
& x^{*} =0, k^{*} =\left\{ \frac{\beta\gamma\exp(x^{*})}{1-(1-\delta)\beta}\right\} ^{\frac{1}{1-\gamma}}, \phi =1+\frac{1}{\beta}+\frac{(1-\gamma)\left[1-(1-\delta)\beta\right]}{\tau}\frac{c^{*}}{k^{*}},\label{equ:k_star}\\
&\frac{c^{*}}{k^{*}}=\frac{\beta^{-1}-1+\delta(1-\gamma)}{\gamma}, q=\beta\left\{ (1-\rho)\left(\frac{c^{*}}{k^{*}}+\delta\right)+\frac{\rho\beta}{\tau}\left(\beta^{-1}-1+\delta\right)\frac{c^{*}}{k^{*}}\right\} k^{*},\notag\\
& \lambda\ \text{is the unique solution such that }\lambda^{2}-\phi\lambda+\frac{1}{\beta}=0\ \text{and}\ |\lambda|\leq 1.\notag
\end{align}}%
%and $c^{*}$ is the corresponding consumption.
Now consider,  instead, a new problem of the finite horizon version as follows
{\allowdisplaybreaks
\begin{align}\label{eq:real_bus_cyc_finite_hor}
 \max_{c_t, 0\leq t\leq T-1}  & E_0\left[\sum_{t=0}^{T}\beta^{t}u(k_{t},k_{t-1},x_{t})\right]=E_0\left[\sum_{t=0}^{T}\beta^{t}\frac{g_{t}^{1-\tau}}{1-\tau}\right]\\
 \text{s.t.}\ \ \  & x_{t+1} =\rho x_{t}+\epsilon_{t+1}, 0\leq t \leq T-1,\notag\\
   & k_{t}=\exp(x_{t})k_{t-1}^{\gamma}-g_{t}+(1-\delta)k_{t-1}, 0\leq t \leq T-1,\notag\\
   & g_{t} = \frac{1}{1+\exp(c_t)}\left[\exp(x_{t})k_{t-1}^{\gamma}+(1-\delta)k_{t-1}\right], 0\leq t \leq T-1,\label{equ:g_t_range_cons}\\
   & g_T = \exp(x_{T})k_{T-1}^{\gamma}+(1-\delta)k_{T-1},\label{eq:terminal_g}\\
   & c_t \in \mathbb{R}, 0\leq t \leq T-1,\notag
\end{align}}%
where \eqref{equ:g_t_range_cons} is used to impose the constraint $0<g_t<\exp(x_{t})k_{t-1}^{\gamma}+(1-\delta)k_{t-1}$ for $t=0,\ldots,T-1$;
\eqref{eq:terminal_g} means that
%Apparently, the optimal solution to \eqref{eq:real_bus_cyc_finite_hor} must satisfy
%$$k_T=0,\ \text{i.e.,}\ g_{T} = \exp(x_{T})k_{T-1}^{\gamma}+(1-\delta)k_{T-1},$$
the available capital at period $T$ is all consumed at period $T$. Hence, in terms of the notation of problem \eqref{equ:multi_per_obj},
the last period utility of problem \eqref{eq:real_bus_cyc_finite_hor} is given by
$$u_T(s_T, s_{T-1}, c_{T-1})=\beta^{T-1}\frac{g_{T-1}^{1-\tau}}{1-\tau}+\beta^{T}\frac{g_{T}^{1-\tau}}{1-\tau},\ \text{where}\ g_T\ \text{is given by \eqref{eq:terminal_g}.}$$
We shall solve this finite time horizon problem by using the EM-C algorithm.

%The Log-linear LQ approximated solution for the problem \eqref{eq:real_bus_cyc_finite_hor},  is given by \eqref{eq:sol_log_linear} for $0\leq t\leq T-1$.
%\begin{eqnarray*}
% & \max_{c_t} & E\left[\sum_{t=0}^{T}\beta^{t}r(k_{t-1},k_{t},x_{t})\right]=E\left[\sum_{t=0}^{T}\beta^{t}\frac{c_{t}^{1-\tau}}{1-\tau}\right]\\
%& \text{s.t.} & x_{t+1} =\rho x_{t}+\epsilon_{t+1}\\
% &  & c_{t}+k_{t}-(1-\delta)k_{t-1}=\exp(x_{t})k_{t-1}^{\gamma}\\
%\mbox{States:} &  & k_{t-1},x_{t} (\mbox{given }k_{0},x_{0})\\
%\mbox{Controls:} &  & c_{t}\in[0,\exp(x_{t})k_{t-1}^{\gamma}+(1-\delta)k_{t-1}]
%\end{eqnarray*}

%\begin{prop}
%Consider $x_{t}\leq M$ (or equivalently $\epsilon_{t+1}\leq M'$)
%for some large constants $M\mbox{ (or }M')\in(0,\infty)$ , then Assumption
%\ref{assump:1} is satisfied.\end{prop}
%\begin{proof}
%For any $k_{t-1}$, feasible set $\Gamma(k_{t-1})=[0,\exp(x_{t})k_{t-1}^{\gamma}+(1-\delta)k_{t-1}]$
%is compact. When $x_{t}\leq M$, $k_{t}$ is bounded from above within
%the whole time horizon. So the objective function is bounded. Assumption
%\ref{assump:1} is satisfied.
%\end{proof}

\subsection{Numerical Results}

Suppose the problem parameters are $\beta=0.98$, $\gamma=0.33$,
$\tau=0.5$, $\delta=0.025$, $\rho=0.95$, and $\epsilon_{t}\overset{d}{\sim} N(0,\sigma_{e}^2)$
with $\sigma_{e}=0.1$. The initial state is $s_0=(k_{-1}, x_{0})=(k^{*}, 0)$, where $k^*$ is given in \eqref{equ:k_star}.
The control $c_t$ is specified as
$$c_t=\sum_{i=1}^4 \theta_{t,i}\phi_i(k_{t-1},x_{t}),$$
where $\{\phi_{i}, i=1, 2, 3, 4\}$ are the basis functions defined as
{\allowdisplaybreaks
\begin{align*}
&\phi_{1}(k_{t-1},x_{t})=1, \phi_{2}(k_{t-1},x_{t})=k_{t-1}, \phi_{3}(k_{t-1},x_{t})=\exp(x_{t}), \phi_{4}(k_{t-1},x_{t})=k_{t-1}^{\gamma}.
\end{align*}}%
In the EM-C algorithm, we initialize $c^0_0 = 0$ and $\theta^0_t = 0$ for all $t$.
We use $N=10,000$ sample paths in the simulation and use $m=2,000$ iterations in the SA algorithm.
%We average $N_{d}=500$ samples for
%gradient estimation as in (\ref{eq:def_Q}).

We first solve the problem \eqref{eq:real_bus_cyc_finite_hor} for the case of 6 years, i.e. $T=6$.
In Figure \ref{fig:Utility_rbc}, cumulative expected utility of EM-C
optimal control and that of the log-linear LQ approximation are illustrated,
based on simulation of $N=10,000$ sample paths. The EM-C
algorithm converges after 3 iterations. It takes about
18 minutes to finish each iteration. The optimal utility obtained by the EM-C algorithm is $28.53$ (with standard error 0.008). The standard error is equal to the sample standard deviation of the $N$ samples on the right-hand side of \eqref{equ:appro_E_0} divided by $\sqrt{N}$.

\begin{figure}[htbp]
\begin{centering}
\includegraphics[width=\textwidth]{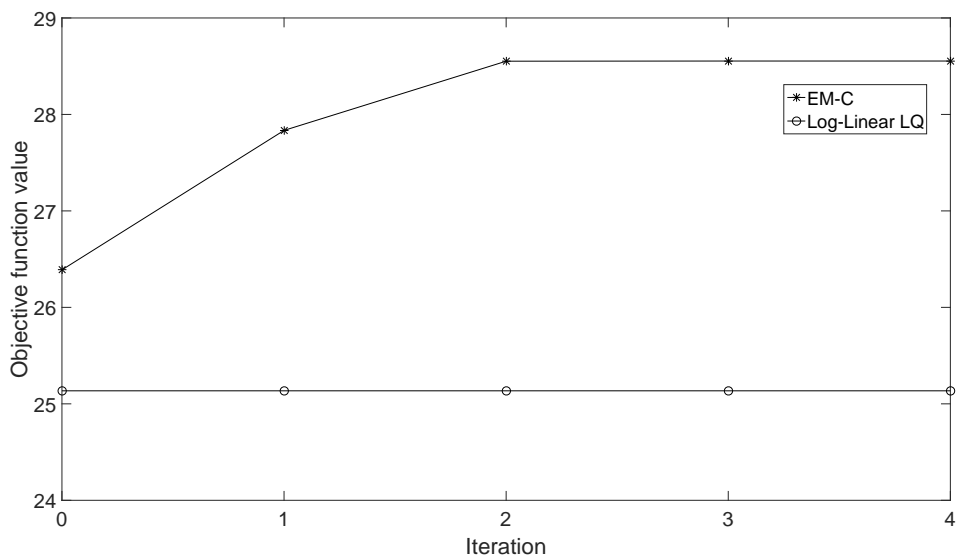}
\end{centering}
\caption{Real business cycle: the comparison of the utility value obtained by the EM-C algorithm and that by the log-linear LQ approximation for the problem \eqref{eq:real_bus_cyc_finite_hor} with $T=6$. The EM-C algorithm converges after 3 iterations. It takes about 18 minutes to finish each iteration. The optimal utility obtained by the EM-C algorithm is $28.53$ (with standard error 0.008). The standard error is equal to the sample standard deviation of the $N$ samples on the right-hand side of \eqref{equ:appro_E_0} divided by $\sqrt{N}$. \label{fig:Utility_rbc}}
\end{figure}

Figure \ref{fig:consumption_cem_vs_LQ} compares the optimal consumption $g_t$ as a function of the state $(k_{t-1}, x_t)$ under the EM-C control for the finite time horizon problem \eqref{eq:real_bus_cyc_finite_hor} and that under the log-linear LQ approach
for the infinite time horizon problem. It is clear from the figure that optimal consumption at period $t=5$ under the EM-C algorithm is much more sensitive to $k_{t-1}$ than that obtained by the log-linear LQ approach.

%Also, the dependence of MC optimal $p$ on residual capacity
%is especially significant for $\theta_{12}^{2}$ and $\theta_{23}^{3}$,
%i.e., $p_{12}$ on $R_{12}$ and $p_{23}$ on $R_{23}$. This is reasonable
%since when closer to the end of the horizon and residual capacity
%is small, the randomness in the arrival process makes the ratios among
%residual capacity deviate farther away from the original ones, hence
%the deterministic optimal control becomes less favorable. As a result,
%a re-optimization becomes necessary. While early in time and residual
%capacity is large, the ratios among residual capacity stays relatively
%stable, hence the deterministic optimal control is a reasonably good
%policy.

\begin{figure}[htbp]
\centering
\includegraphics[width=\textwidth]{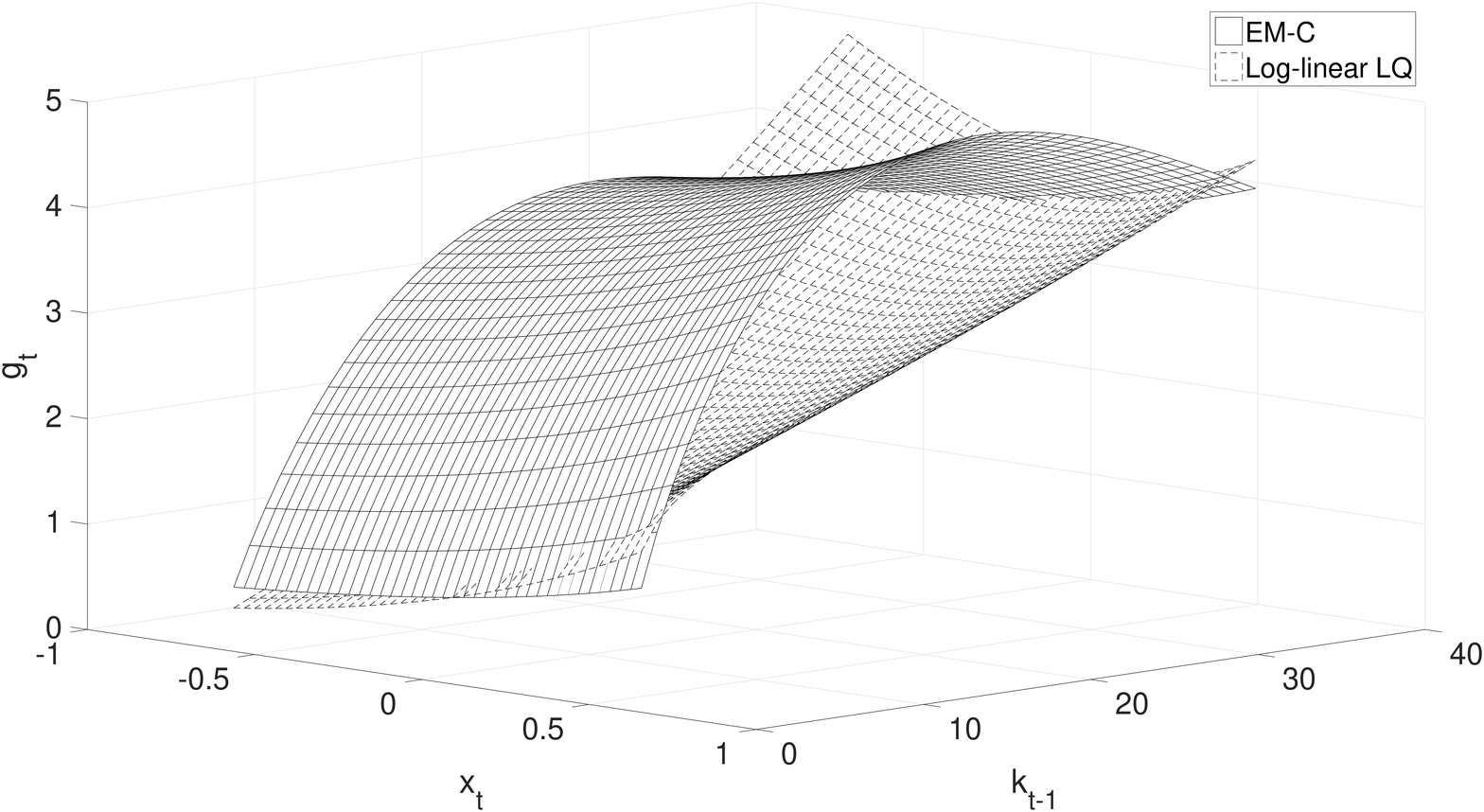}
\centering
\includegraphics[width=\textwidth]{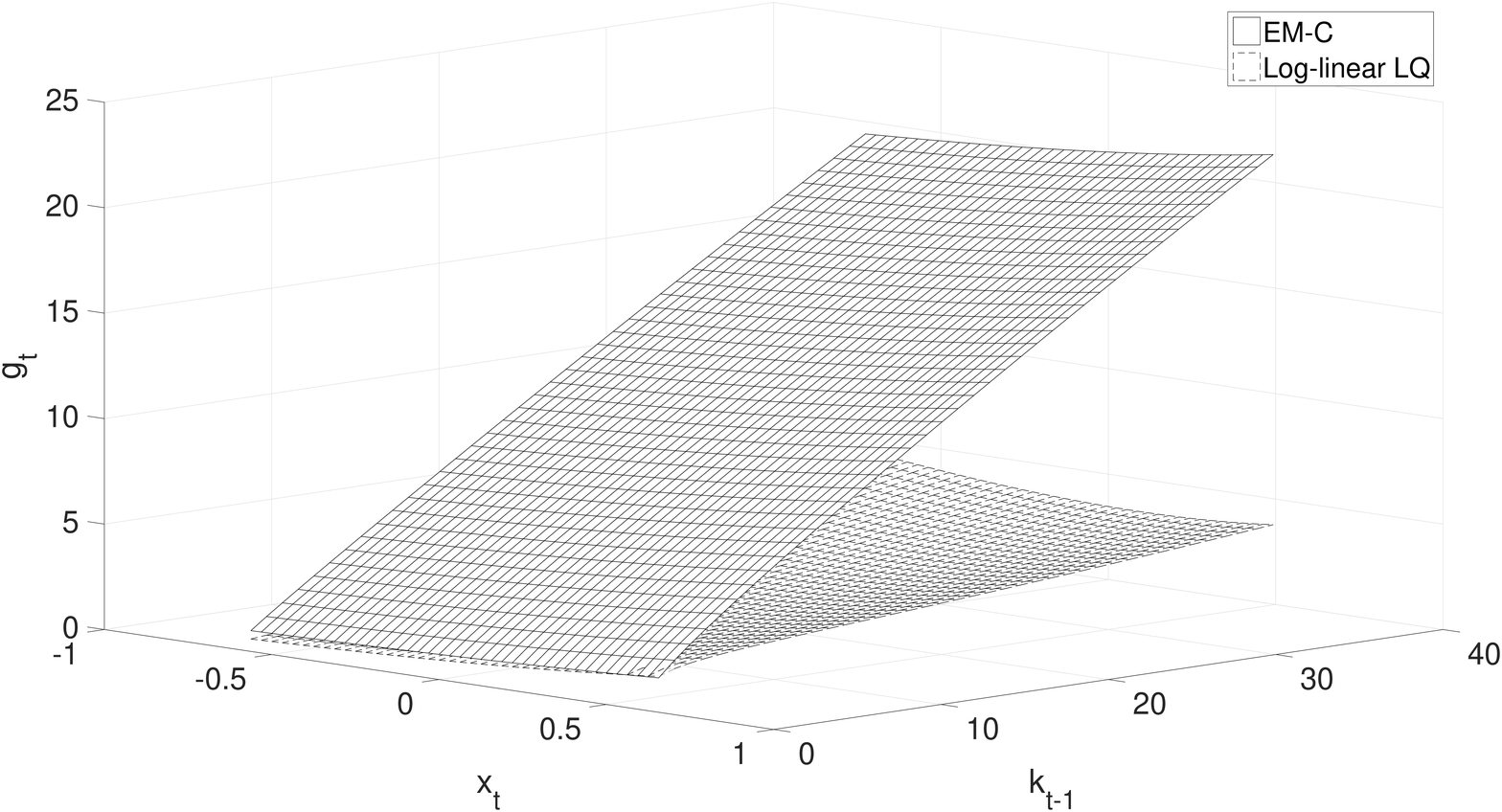}
\caption{Real business cycle problem: the comparison of the optimal consumption (control policy) $g_t$ as a function of the state $(k_{t-1}, x_t)$ under the EM-C control for the finite time horizon problem \eqref{eq:real_bus_cyc_finite_hor} with $T=6$
and that under the log-linear LQ approach
for the infinite time horizon problem.
The top figure plots $g_t$ for $t=2$, and the bottom one plots $g_t$ for $t=5$, which is the second to the last period.}
%Each plot shows the price for certain itinerary with respect
%to residual capacities.
\label{fig:consumption_cem_vs_LQ}
\end{figure}

We then solve the problem \eqref{eq:real_bus_cyc_finite_hor} for the case of 10 years, i.e., $T=10$.
In Figure \ref{fig:Utility_rbc_T_10}, cumulative expected utility of EM-C
optimal controls and that of the log-linear LQ approximation are illustrated,
based on simulation of $N=10,000$ sample paths. The EM-C
algorithm converges after 3 iterations. It takes about 30 minutes to finish each iteration. The optimal utility obtained by the EM-C algorithm is $38.04$ (with standard error 0.016). The standard error is equal to the sample standard deviation of the $N$ samples on the right-hand side of \eqref{equ:appro_E_0} divided by $\sqrt{N}$.

\begin{figure}[htbp]
\begin{centering}
\includegraphics[width=\textwidth]{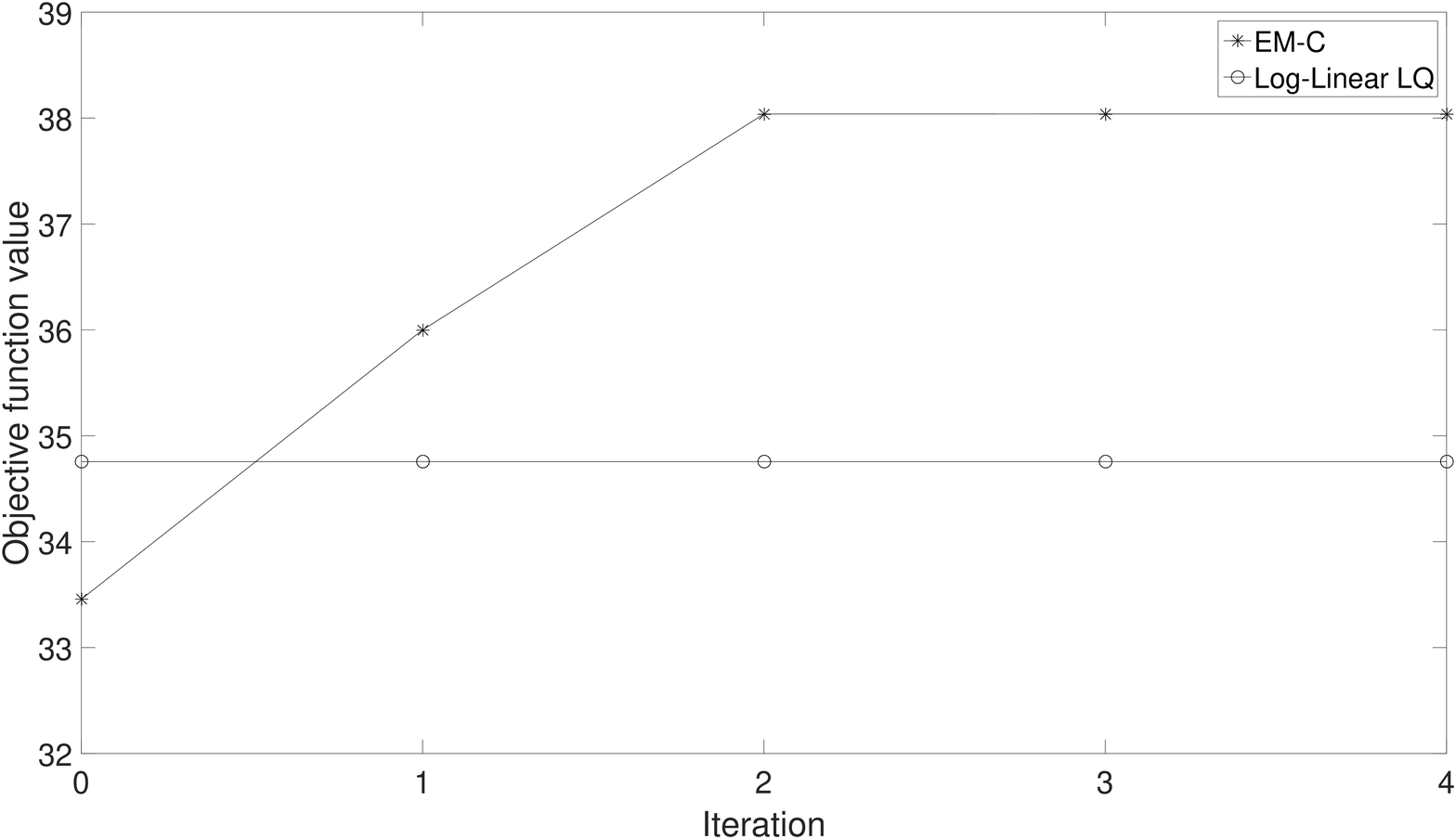}
\end{centering}
\caption{Real business cycle: the comparison of the utility value obtained by the EM-C algorithm and that by the log-linear LQ approximation for the problem \eqref{eq:real_bus_cyc_finite_hor} with $T=10$. The EM-C algorithm converges after 3 iterations. It takes about 30 minutes to finish each iteration. The optimal utility obtained by the EM-C algorithm is $38.04$ (with standard error 0.016). The standard error is equal to the sample standard deviation of the $N$ samples on the right-hand side of \eqref{equ:appro_E_0} divided by $\sqrt{N}$. \label{fig:Utility_rbc_T_10}}
\end{figure}

Figure \ref{fig:consumption_cem_vs_LQ_T_10} compares the optimal consumption (control policy)  $g_t$ as a function of the state $(k_{t-1}, x_t)$ under the EM-C control for the problem \eqref{eq:real_bus_cyc_finite_hor} with $T=10$ and that under the log-linear LQ approach
for the infinite time horizon problem. It is clear from the figure that optimal consumption at period $t=9$ under the EM-C algorithm is much more sensitive to $k_{t-1}$ than that obtained by the log-linear LQ approach.

%Also, the dependence of MC optimal $p$ on residual capacity
%is especially significant for $\theta_{12}^{2}$ and $\theta_{23}^{3}$,
%i.e., $p_{12}$ on $R_{12}$ and $p_{23}$ on $R_{23}$. This is reasonable
%since when closer to the end of the horizon and residual capacity
%is small, the randomness in the arrival process makes the ratios among
%residual capacity deviate farther away from the original ones, hence
%the deterministic optimal control becomes less favorable. As a result,
%a re-optimization becomes necessary. While early in time and residual
%capacity is large, the ratios among residual capacity stays relatively
%stable, hence the deterministic optimal control is a reasonably good
%policy.

\begin{figure}[htbp]
\centering
\includegraphics[width=\textwidth]{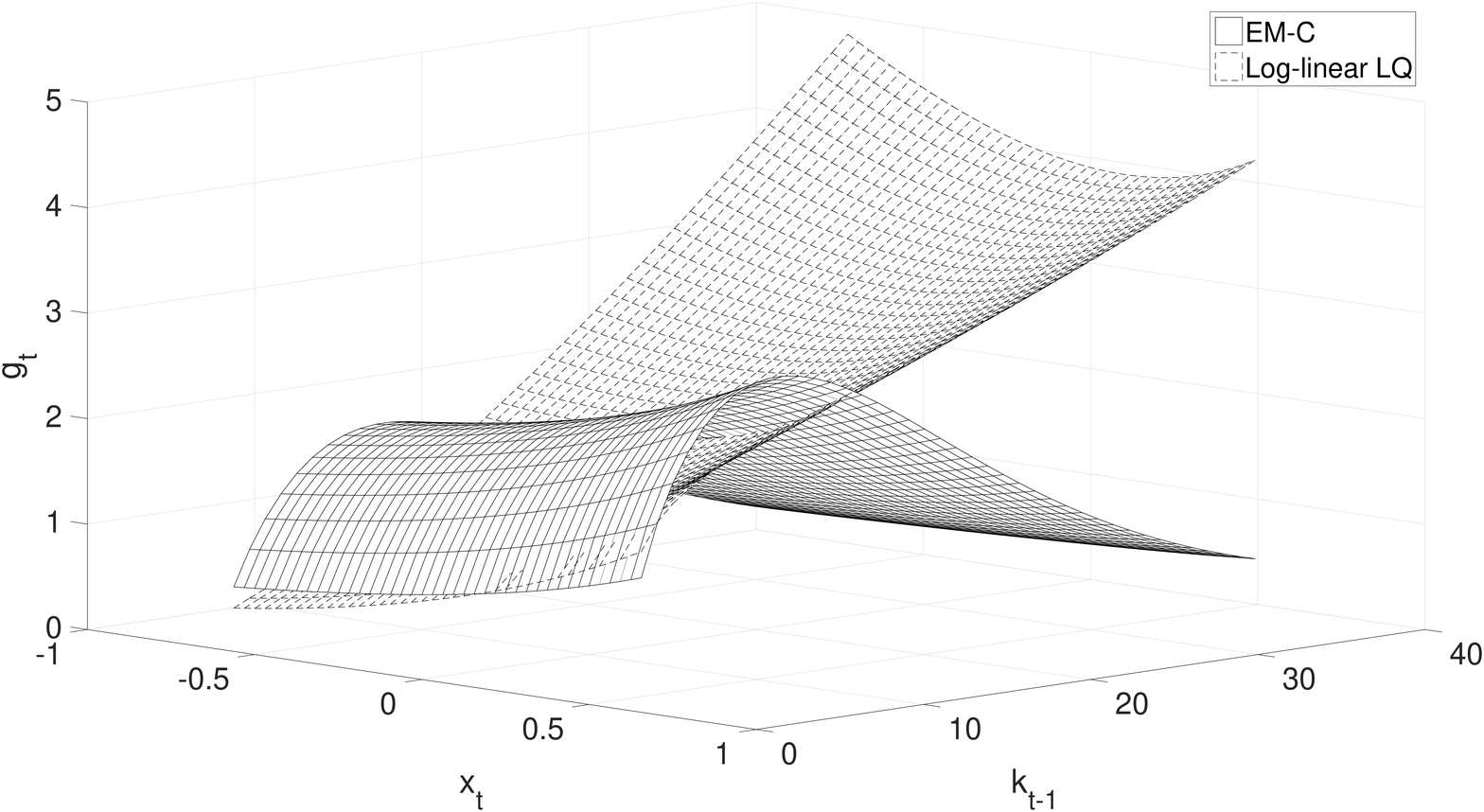}
\centering
\includegraphics[width=\textwidth]{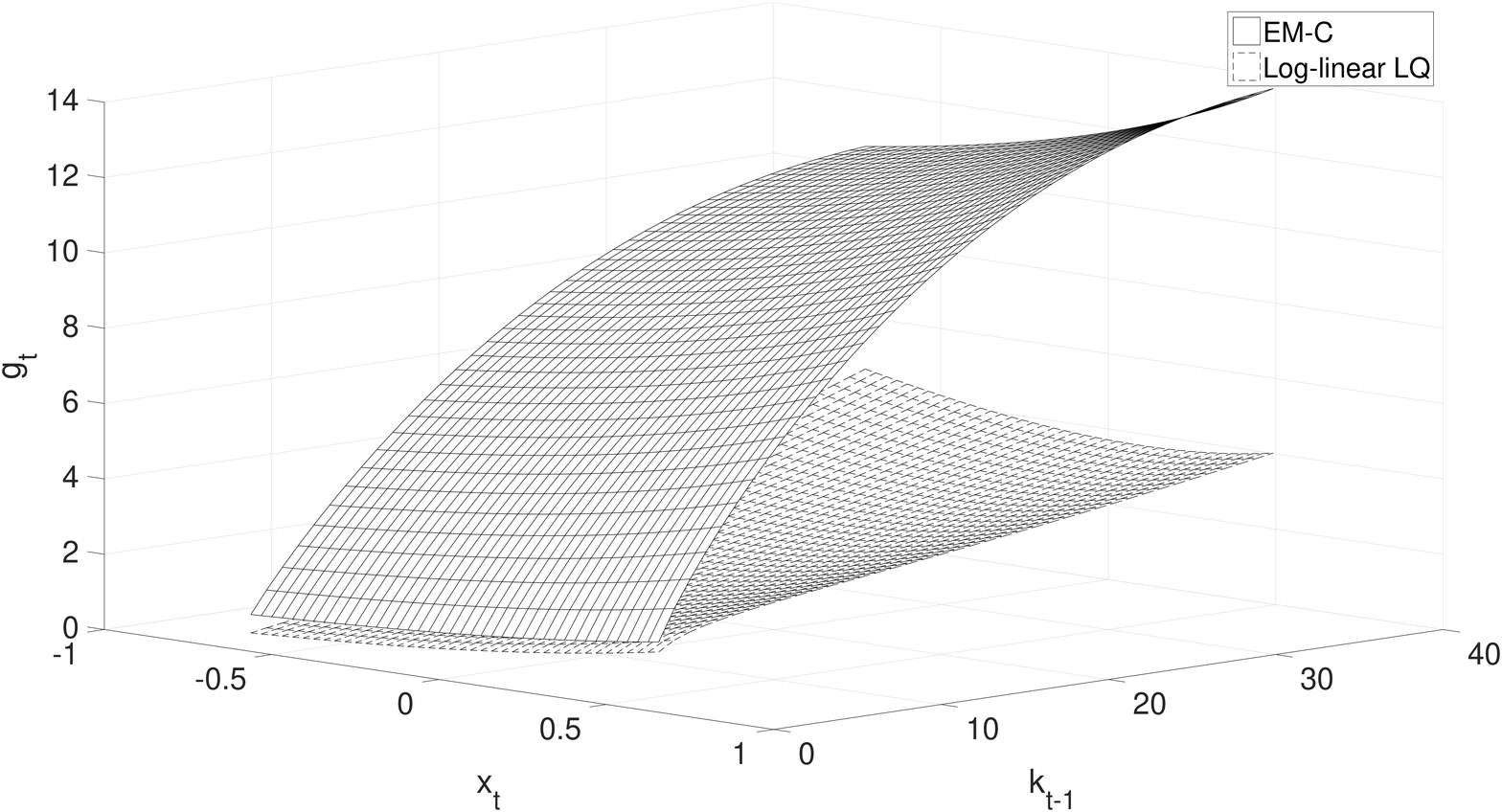}
\caption{Real business cycle problem: the comparison of the optimal consumption $g_t$ as a function of the state $(k_{t-1}, x_t)$ under the EM-C control for the problem \eqref{eq:real_bus_cyc_finite_hor} with $T=10$ and that under the log-linear LQ approach
for the infinite time horizon problem. The top figure plots $g_t$ for $t=2$, and the bottom one plots $g_t$ for $t=9$, which is the second to the last period.}
%Each plot shows the price for certain itinerary with respect
%to residual capacities.
\label{fig:consumption_cem_vs_LQ_T_10}
\end{figure}

\appendix
\begin{appendices}

%\section{Stochastic Approximation Algorithm for Solving Problems \eqref{eq:cem_opt_T_1_S}, \eqref{eq:opt_t_S}, and \eqref{eq:opt_0}}\label{sec:Stochastic-Approximation}

\section{A Simple Derivation}\label{app:simple_deriv}
We will show that \eqref{eq:monot} is equivalent to \eqref{equ:new_2}. In fact, by \eqref{equ:utility_func}, \eqref{eq:monot} is equivalent to
{\allowdisplaybreaks
\begin{align}\label{equ:new_1}
&  E_0\left[\sum_{j=0}^{t-1} u_{j+1}(s_{j+1}, s_{j}, c_j)\right.\notag\\
                    &\quad\quad\quad\quad\quad \left.+\sum_{j=t}^{T-1} u_{j+1}(s_{j+1},s_{j},c_j)\middle | c_0^{k-1}, \theta_1^{k-1}, \ldots, \theta_{t-1}^{k-1}, \theta_t^k, \theta_{t+1}^k, \ldots, \theta_{T-1}^k \right]\notag\\
\geq{} &  E_0\left[\sum_{j=0}^{t-1} u_{j+1}(s_{j+1}, s_{j}, c_j)\right.\notag\\
                    &\quad\quad\quad\quad\quad \left.+\sum_{j=t}^{T-1} u_{j+1}(s_{j+1},s_{j},c_j)\middle | c_0^{k-1}, \theta_1^{k-1}, \ldots, \theta_{t-1}^{k-1}, \theta_t^{k-1}, \theta_{t+1}^k, \ldots, \theta_{T-1}^k \right].
\end{align}}%
By \eqref{eq:c_t_s_t} and \eqref{eq:state_evo}, $\sum_{j=0}^{t-1} u_{j+1}(s_{j+1}, s_{j}, c_j)$ depends on the control parameters $(c_0, \theta_1, \ldots, \theta_{t-1})$ but not on the control parameter $(\theta_{t}, \ldots, \theta_{T-1})$. Therefore, we have
{\allowdisplaybreaks
\begin{align*}
&  E_0\left[\sum_{j=0}^{t-1} u_{j+1}(s_{j+1}, s_{j}, c_j)\middle | c_0^{k-1}, \theta_1^{k-1}, \ldots, \theta_{t-1}^{k-1}, \theta_t^k, \theta_{t+1}^k, \ldots, \theta_{T-1}^k \right]\\
={} &  E_0\left[\sum_{j=0}^{t-1} u_{j+1}(s_{j+1}, s_{j}, c_j)\middle | c_0^{k-1}, \theta_1^{k-1}, \ldots, \theta_{t-1}^{k-1}, \theta_t^{k-1}, \theta_{t+1}^k, \ldots, \theta_{T-1}^k \right],
\end{align*}}%
which implies that \eqref{equ:new_1} is equivalent to \eqref{equ:new_2}.

\section{Proof of Theorems}
\subsection{Proof of Theorem \ref{thm:monoto}}\label{app:proof_monoto}
\begin{proof}
%In the EM-C algorithm, the iterations satisfy \eqref{eq:opt_T_1_U}, \eqref{eq:opt_t_U}, and \eqref{eq:opt_t_S}.
In the EM-C algorithm, the iterations satisfy %\eqref{eq:mono_T_minus_1},
\eqref{eq:monot} and \eqref{eq:mono_0}. Therefore, we have
{\allowdisplaybreaks
  \begin{align}
    & \phantom{{}={}} U(c^{k-1}_0, \theta^{k-1}_1, \theta^{k-1}_2, \ldots, \theta^{k-1}_{T-3}, \theta^{k-1}_{T-2}, \theta^{k-1}_{T-1})\notag\\
&\leq  U(c^{k-1}_0, \theta^{k-1}_1, \theta^{k-1}_2, \ldots, \theta^{k-1}_{T-3}, \theta^{k-1}_{T-2}, \theta^{k}_{T-1})\notag\\
&\leq  U(c^{k-1}_0, \theta^{k-1}_1, \theta^{k-1}_2, \ldots, \theta^{k-1}_{T-3}, \theta^{k}_{T-2}, \theta^{k}_{T-1})\notag\\
&\leq  \cdots \notag\\
&\leq  U(c^{k-1}_0, \theta^{k}_1, \theta^{k}_2, \ldots, \theta^{k}_{T-3}, \theta^{k}_{T-2}, \theta^{k}_{T-1})\notag\\
&\leq  U(c^{k}_0, \theta^{k}_1, \theta^{k}_2, \ldots, \theta^{k}_{T-3}, \theta^{k}_{T-2}, \theta^{k}_{T-1}),\notag
  \end{align}}%
from which the proof is completed.
%  By definition of $\theta_t^k$, the value function increase when $\theta_t^{k-1}$ is replaced by $\theta_t^k$.
\end{proof}

\subsection{Proof of Theorem \ref{thm:convergence_red_em}}\label{app:proof_red_em}
\begin{proof}
We first recall the following definition in \citet{wu1983convergence}: A point-to-set map $\rho$ on $X$ is said to be closed at $x$, if $x^k \to x$, $x^k\in X$, $y^k \to y$, and $y^k\in \rho(x^k)$ imply $y\in \rho(x)$.
%If $\rho$ is a point-to-point map and $\rho$ is continuous, then $\rho$ is closed.
We also recall the following global convergence theorem (\citet[][p. 91]{Zangwill1969}): Let the sequence $\{x^k\}_{k=0}^{\infty}$ be generated by $x^k\in M(x^{k-1})$, where $M$ is a point-to-set map on $X$. Let a solution set $\Gamma \subset X$ be given, and suppose that: (i) all points $x^k$ are contained in a compact set $S \subset X$; (ii) $M$ is closed over the complement of $\Gamma$; (iii) there is a continuous function $\alpha$ on $X$ such that (a) if $x \notin \Gamma$, $\alpha(y) > \alpha(x)$ for all $y \in M(x)$, and (b) if $x \in \Gamma$, $\alpha(y) \geq \alpha(x)$ for all $y \in M(x)$. Then all the limit points of ${x^k}$ are in the solution set $\Gamma$ and $\alpha(x^k)$ converges monotonically to $\alpha(x^*)$ for some $x^*\in \Gamma$.

We now prove part (1) of the theorem. First, we show that $M$ is a closed point-to-set map on $\R^n$. Suppose
$$a^k=(a_0^k, a_1^k, \ldots, a^k_{T-1})\to \bar a=(\bar a_0, \bar a_1, \ldots, \bar a_{T-1}),\ \text{as}\ k\to\infty.$$
Suppose $b^k=(b_0^k, b_1^k, \ldots, b^k_{T-1})\in M(a^k)$ and $b^k\to \bar b=(\bar b_0, \bar b_1, \ldots, \bar b_{T-1})$ as $k\to\infty$. We will show that $\bar b\in M(\bar a)$. Since $b^k\in M(a^k)$, it follows that
{\allowdisplaybreaks
\begin{align}
&U(a_0^k, a_1^k, \ldots, a^k_{T-2}, b^k_{T-1})\geq U(a_0^k, a_1^k, \ldots, a^k_{T-2}, a^k_{T-1}), \forall k\notag\\
&U(a_0^k, a_1^k, \ldots, a^k_{t-1}, b^k_{t}, b^k_{t+1}, \ldots, b^k_{T-1})\geq U(a_0^k, a_1^k, \ldots, a^k_{t-1}, a^k_{t}, b^k_{t+1}, \ldots, b^k_{T-1}), \forall t, \forall k\notag\\
&U(b_0^k, b_1^k, \ldots, b^k_{T-1})\geq U(a_0^k, b_1^k, \ldots, b^k_{T-1}), \forall k.\notag
\end{align}}%
Letting $k\to \infty$ in the above inequalities, we obtain from the continuity of $U$ that
{\allowdisplaybreaks
\begin{align}
&U(\bar a_0, \bar a_1, \ldots, \bar a_{T-2}, \bar b_{T-1})\geq U(\bar a_0, \bar a_1, \ldots, \bar a_{T-2}, \bar a_{T-1}), \forall k\notag\\
&U(\bar a_0, \bar a_1, \ldots, \bar a_{t-1}, \bar b_{t}, \bar b_{t+1}, \ldots, \bar b_{T-1})\geq U(\bar a_0, \bar a_1, \ldots, \bar a_{t-1}, \bar a_{t}, \bar b_{t+1}, \ldots, \bar b_{T-1}), \forall t, \forall k\notag\\
&U(\bar b_0, \bar b_1, \ldots, \bar b_{T-1})\geq U(\bar a_0, \bar b_1, \ldots, \bar b_{T-1}), \forall k,\notag
\end{align}}%
which implies that $\bar b\in M(\bar a)$. Hence, $M$ is a closed point-to-set map on $\R^n$.

Second, we will verify that the conditions of the global convergence theorem cited above hold. Let $\alpha(x)$ be $U(x)$ and the solution set $\Gamma$ to be $\mathcal{S}$ or $\mathcal{M}$. Then, condition (i) follows from \eqref{equ:assump_1} and \eqref{equ:monoto}. Condition (ii) has been approved above. Condition (iii) (a) follows from \eqref{equ:cond_converg}. Condition (iii) (b) follows from \eqref{equ:monoto}. Hence, the conclusion of part (1) of the theorem follows from the global convergence theorem.

We move to prove part (2) of the theorem. To prove part (2), we only need to show that, under the condition of part (2), \eqref{equ:cond_converg} holds for any $x^{k-1}\notin \mathcal{S}$. For any such $x^{k-1}$, it follows from
the definition of the set $\mathcal{S}$ that
$\frac{\partial U(x^{k-1})}{\partial x^{k-1}}\neq 0$. Suppose $x^k = x^{k-1}$. Then, for each $j=T-1, T-2, \ldots, 1, 0$, $x_j^{k-1}$ maximizes the function $H_j(y):=U(x_0^{k-1}, x_1^{k-1}, \ldots, x_{j-1}^{k-1}, y, x_{j+1}^{k-1}, \ldots, x_{T-1}^{k-1})$, which implies that $\frac{\partial U(x^{k-1})}{\partial x_j^{k-1}}=0$ for all $j$, which contradicts to that $\frac{\partial U(x^{k-1})}{\partial x^{k-1}}\neq 0$. Hence, $x^k \neq x^{k-1}$. Let $i_0$ be the largest index $j\in\{0, 1, \ldots, T-1\}$ such that $x^k_j\neq x^{k-1}_j$. Then, by the specification of the algorithm,
%Therefore, there exists some index $0\leq i\leq T-1$ such that $\frac{\partial U(x^{k-1})}{\partial x_i^{k-1}}\neq 0$. Let $i_0$ be the largest of all such indices. Then, since $\frac{\partial U(x^{k-1})}{\partial x_{i_0}^{k-1}}\neq 0$,
$x_{i_0}^{k}$ maximizes the function $H_{i_0}(y):=U(x_0^{k-1}, x_1^{k-1}, \ldots, x_{i_0-1}^{k-1}, y, x_{i_0+1}^{k-1}, \ldots, x_{T-1}^{k-1})$ but $x_{i_0}^{k-1}$ does not. Hence,
%Let $y_{i_0}^{k-1}=\arg \max H(y)$. Then, by the definition of the EM-C algorithm, $x_{i_0-1}^k$ is updated to be $y_{i_0}^k$, and
$$H_{i_0}(x_{i_0}^k)>H_{i_0}(x_{i_0}^{k-1})=U(x^{k-1}),$$
which implies that
$$U(x^{k})\geq H_{i_0}(x_{i_0}^k)>U(x^{k-1}).$$
Hence, \eqref{equ:cond_converg} holds for any $x^{k-1}\notin \mathcal{S}$ for the EM-C algorithm. Then, the conclusion of part (2) follows from part (1) of the theorem, which has been proved.
\end{proof}

\subsection{Proof of Theorem \ref{thm:converg_x}}\label{app:proof_cong_x}
\begin{proof}
We first prove part (1). By Theorem \ref{thm:convergence_red_em}, all the limit points of $\{x^k\}_{k\geq 0}$ are in $\cS(U^*)=\{x^*\}$ (resp. $\M(U^*)=\{x^*\}$). Hence, any converging subsequence of $\{x^k\}_{k\geq 0}$ converges to $x^*$, which implies that $x^k\to x^*$ as $k\to\infty$. Hence, part (1) of the theorem holds.
Next we prove part (2). By the condition \eqref{equ:assump_1}, $\{x^k\}$ is a bounded sequence. By Theorem 28.1 of \citet*{Ostrowski-1966}, the set of limit points of the bounded sequence $\{x^k\}$ with $\|x^{k+1}-x^k\|\to 0$ as $k\to\infty$ is compact and connected. In addition, by Theorem \ref{thm:convergence_red_em}, all the limit points of $\{x^k\}$ are in $\cS(U^*)$ (resp. $\M(U^*)$). Hence, the conclusion of part (2) follows.
\end{proof}

\section{Stochastic Approximation Algorithm for Solving Problems \eqref{eq:opt_t_S} and \eqref{eq:opt_0}}\label{sec:Stochastic-Approximation}

By %\eqref{equ:simu_E_T_1},
\eqref{equ:simu_E_t} and \eqref{eq:simu_E_0}, the problems %\eqref{eq:cem_opt_T_1_S},
\eqref{eq:opt_t_S} and \eqref{eq:opt_0} have the general form
\begin{equation}\label{eq:sa_problem_general}
        \max_{y\in \Upsilon}
       E_{0}\left[\tilde f(y)\right],
\end{equation}
where $\Upsilon\subset \mathbb{R}^m$, $\tilde f(\cdot)$ is defined in %\eqref{equ:appro_E_T_1},
\eqref{equ:appro_E_t} and \eqref{equ:appro_E_0}, respectively.
Let $\{a^k=(a^k_1, \ldots, a^k_m)\}_{k=1}^{\infty}$ and $\{b^k=(b^k_1, \ldots, b^k_m)\}_{k=1}^{\infty}$ be two deterministic vector sequences such that
{\allowdisplaybreaks
\begin{align*}
& a^k > 0,\ b^k > 0,\forall k,\\
& a^{k} \rightarrow 0,\ b^{k}\rightarrow 0,\ \text{as}\ k\to\infty,\\
& \sum_{k=1}^{\infty}a^{k}_i=\infty,\ \sum_{k=1}^{\infty}\frac{(a^{k}_i)^{2}}{(b^{k}_i)^{2}}<\infty,\ \text{as}\ k\to\infty,\ \forall i=1,\ldots, m.
\end{align*}}%
Let $\delta_i=(0,\ldots, 0, 1, 0, \ldots, 0)'$ be the $i$th standard basis of $\mathbb{R}^m$.

The SA algorithm for solving the problem \eqref{eq:sa_problem_general} is then given by
\begin{enumerate}
  \item
  Initialize $y^1\in \mathbb{R}^m$ and $k=1$.
  \item
  Iterate $k$ until some stopping criteria are met. At the $(k+1)$th iteration, update $y^{k}$ to be
  \be
  y^{k+1}_i = y^{k}_i + a^{k}_i \left(\frac{\tilde f(y^{k}+c^{k}_i\delta_i) - \tilde f(y^{k}-c^{k}_i\delta_i) }{c^{k}_i}\right),\ i=1,\ldots, m.\notag
  \ee
\end{enumerate}

%Note in (\ref{eq:alg_sa}), $\mbox{ }Q^{r}(s_{t,l}^{k},c_{t}^{k,n-1})$
%is the cumulative objective of sample paths starting from $s_{t,l}^{k}$
%until the end of the horizon with policy $\left\{ c_{t}^{k,n-1},\hat{c}_{t+1}, \ldots, \hat{c}_{T}\right\} $,
%and $D(s_{t,l}^{k},c_{t}^{k,n-1},b_{n})$ is the finite-difference
%gradient estimator with respect to the policy at time $t$. The evaluation
%of $Q^{r}(s_{t,l}^{k},c_{t}\pm c_{n})$ requires to simulate sample
%paths of starting from $s_{t,l}^{k}$.
%
%We can increase the accuracy of the gradient estimator, $D(s_{t,l}^{k},c_{t}^{k,n-1},c_{n})$,
%by using various variance reduction techniques. First, we use the
%average of $N_{d}$ number of $Q^{r}$ to reduce the noise. In some
%of our numerical examples, we construct control-variate estimator
%for $D(s_{t,l}^{k},c_{t}^{k,n-1},c_{n})$, which gives better convergence
%speed for SA.

To reduce variance in the SA algorithm, at each iteration $k$, we use the common random numbers for generating the $2m$ random variables
$\tilde f(y^{k}+c^{k}_i\delta_i)$ and $\tilde f(y^{k}-c^{k}_i\delta_i)$, $i=1,\ldots, m$.

In the numerical examples of this paper, we used the scaled-and-shifted stochastic approximation
(SSSA) algorithm in \citet*{broadie2011general}, where the sequence $a^k$ and $b^k$ are chosen as $a^{k}=k^{-1}\cdot a^0$
and $b^{k}=k^{-1/4}\cdot b^0$, where $a^0$ and $b^0$ are some initial vector, usually chosen to be proportional to the scale of $y$.

\section{The EM-C Algorithm for the General Control Problem \eqref{equ:multi_per_obj_gen}}\label{app:CEM-general-control}

The EM-C algorithm also works for the general control problem \eqref{equ:multi_per_obj_gen} in which the utility function may not be time-separable. In such problems, %\eqref{eq:mono_T_minus_1} can no longer be simplified to be \eqref{eq:cem_update_TminusOne}, neither can
\eqref{eq:monot} can no longer be simplified to be \eqref{equ:new_2}. To make the EM-C algorithm work for such problems, at each $t$th period of iteration $k$, one just need to
%set $\theta^{k}_{T-1}$ as a suboptimal (optimal) solution of
%\begin{equation}\label{eq:cem_opt_T_1_S_gen}
%        %\theta_{T-1}^{k} = \arg
%        \max_{\theta_{T-1}\in \R^{d}}
%       E_{0}\left[u(s_0, c_0, s_1, c_1, \ldots, s_{T-1}, c_{T-1}, s_{T})\middle | c_0^{k-1}, \theta_1^{k-1}, \ldots, \theta_{T-2}^{k-1}, \theta_{T-1}\right],
%\end{equation}
%and
set $\theta^{k}_{t}$ as a suboptimal (optimal) solution to
            \begin{align}\label{eq:opt_t_S_gen}
              %\theta_{t}^{k} = \arg
       &    \max_{\theta_{t}\in \Theta_t}
                    E_0\left[u(s_0, c_0, s_1, c_1, \ldots, s_{T-1}, c_{T-1}, s_{T})\middle | c_0^{k-1}, \theta_1^{k-1}, \ldots, \theta_{t-1}^{k-1}, \theta_t, \theta_{t+1}^k, \ldots, \theta_{T-1}^k\right].
            \end{align}

%  To solve the general control problem \eqref{equ:multi_per_obj_gen}, one needs to solve the subproblems %\eqref{eq:cem_opt_T_1_S_gen} and
%  \eqref{eq:opt_t_S_gen}.
%  The objective function in \eqref{eq:cem_opt_T_1_S_gen} is equal to
%{\allowdisplaybreaks
%\begin{align}\label{equ:simu_E_T_1_gen}
%&E_{0}\left[u(s_0, c_0, s_1, c_1, \ldots, s_{T-1}, c_{T-1}, s_{T})\middle | c_0^{k-1}, \theta_1^{k-1}, \ldots, \theta_{T-2}^{k-1}, \theta_{T-1}\right]\notag\\
%%=& E_{0}\left[E_{T-1}\left[u_{T}(s_{T}, s_{T-1}, c_{T-1})\mid s_{T-1}\right]\right]\\
%={} & E_0\left[\frac{1}{N}\sum_{l=1}^N
%u(s_0, c_0, s^k_{1,l}, c^k_{1,l}, \ldots, s^k_{T-1,l}, c^k_{T-1,l}, s^k_{T,l}(\theta_{T-1}))\right].
%\end{align}}%
%Then, a SA algorithm can use
%\be\label{equ:appro_E_T_1_gen}
%\tilde f(\theta_{T-1}):=\frac{1}{N}\sum_{l=1}^N
%u(s_0, c_0, s^k_{1,l}, c^k_{1,l}, \ldots, s^k_{T-1,l}, c^k_{T-1,l}, s^k_{T,l}(\theta_{T-1}))
%\ee
%as an approximation to the objective function when solving the subproblem \eqref{eq:cem_opt_T_1_S_gen}. Similarly, one can apply a SA algorithm to solve the subproblem \eqref{eq:opt_t_S_gen}.

The convergence theorems \ref{thm:monoto}, \ref{thm:convergence_red_em}, and \ref{thm:converg_x} also hold for the EM-C algorithm for the general control problem \eqref{equ:multi_per_obj_gen}.

In the implementation of the EM-C algorithm for the general control problem \eqref{equ:multi_per_obj_gen}, one needs to solve the subproblem %\eqref{eq:cem_opt_T_1_S_gen} and
\eqref{eq:opt_t_S_gen}, where
the objective function is:
{\allowdisplaybreaks
\begin{align}%\label{equ:simu_E_t_gen}
&E_{0}\left[u(s_0, c_0, s_1, c_1, \ldots, s_{T-1}, c_{T-1}, s_{T})\middle | c_0^{k-1}, \theta_1^{k-1}, \ldots, \theta_{t-1}^{k-1}, \theta_t, \theta_{t+1}^k, \ldots, \theta_{T-1}^k\right]\notag\\
%=& E_{0}\left[E_{T-1}\left[u_{T}(s_{T}, s_{T-1}, c_{T-1})\mid s_{T-1}\right]\right]\\
={} & E_0\left[\frac{1}{N}\sum_{l=1}^N
u(s_0, c_0, s^k_{1,l}, c^k_{1,l}, \ldots, s^k_{t,l}, c^k_{t,l}(\theta_t), \ldots, s^k_{T-1,l}(\theta_{t}), c^k_{T-1,l}(\theta_t), s^k_{T,l}(\theta_{t}))\right].\notag
\end{align}}%
Then, a SA algorithm can use
\begin{equation*}%\label{equ:appro_E_t_gen}
\tilde f(\theta_{t}):=\frac{1}{N}\sum_{l=1}^N
u(s_0, c_0, s^k_{1,l}, c^k_{1,l}, \ldots, s^k_{t,l}, c^k_{t,l}(\theta_t), \ldots, s^k_{T-1,l}(\theta_{t}), c^k_{T-1,l}(\theta_t), s^k_{T,l}(\theta_{t}))
\end{equation*}
as an approximation to the objective function when solving the subproblem \eqref{eq:opt_t_S_gen}.

\end{appendices}

\bibliographystyle{dcu} %{agsm}
\bibliography{dynamic_EM_ref}

%\pagebreak{}

\end{document}